\documentclass[reprint,amsmath,amssymb,aps,pre,showkeys,showpacs,floats,nofootinbib]{revtex4-1}

\usepackage{frcursive}
\usepackage{enumitem}

\usepackage{arrayjobx}
\newcommand{\eqdef}{\stackrel{\mathrm{def}}{=}} 

\usepackage[table]{xcolor}
\usepackage{multirow}
\usepackage{citesort}
\usepackage{graphics}
\usepackage{graphicx}
\usepackage{wrapfig}
\usepackage[center,normalsize]{subfigure}
\usepackage{amssymb}
\usepackage{amsmath}
\usepackage{xspace}
\usepackage{url}
\usepackage{pifont, epsfig, rotating}
\usepackage{color}
\usepackage{nicefrac}
\usepackage{textcomp}
\usepackage{mathrsfs}
\usepackage{pstricks,pst-node,pst-plot,pst-coil,pst-text,multido,pst-3d,pst-tree,pst-grad}
\usepackage{txfonts}
\usepackage{dcolumn}
\usepackage{longtable}

\definecolor{dgreyblue}{rgb}{0.26,0.3,0.46}             

\newcommand{\cC}{\mathcal{C}}

\newcommand{\cD}{\mathcal{D}}

\renewcommand{\text}[1]{\hbox{\rm \ #1\ \/}}

\newcommand{\beqn}{\begin{eqnarray*}}
\newcommand{\eeqn}{\end{eqnarray*}}
\newcommand{\beq}{\begin{eqnarray}}
\newcommand{\eeq}{\end{eqnarray}}
\newcommand{\ben}{\begin{enumerate}}
\newcommand{\een}{\end{enumerate}}
\newcommand{\bi}{\begin{itemize}}
\newcommand{\ei}{\end{itemize}}
\newcommand{\eps}{\varepsilon}

\newcommand{\gadd}{\delta^{+}}

\newcommand{\IE}{{\em i.e.}\xspace}
\newcommand{\tx}{^{\rm th}}

\newtheorem{theorem}{Theorem}
\newtheorem{lemma}[theorem]{Lemma}

\newtheorem{corollary}[theorem]{Corollary}

\newenvironment{proof}{{\noindent\bf Proof.\ }}{\hfill{\Pisymbol{pzd}{113}}\vspace{0.1in}}
\newenvironment{proof-sketch}{{\noindent\bf Sketch of Proof.\ }}{\hfill{\Pisymbol{pzd}{113}}\vspace{0.1in}}

\newcommand{\NP}{\mathsf{NP}}

\newcommand{\cS}{\mathcal{S}}
\newcommand{\cP}{\mathcal{P}}
\newcommand{\cQ}{\mathcal{Q}}

\newcommand{\TB}{\vspace{-0.1ex}}\newcommand{\TiE}{\setlength{\itemsep}{-1ex}}

\newcommand{\comment}[1]{}

\newcommand{\eX}[1]{{\mathbb E}\left[#1\right]}
\newcommand{\EG}{{\it e.g.}\xspace}
\newcommand{\FI}[1]{Fig.~\ref{#1}\xspace}

\newcommand{\Nbr}{\mathsf{Nbr}}

\newcommand{\gscaled}{\delta^{\mathrm{Y}}}
\newcommand{\Gscaled}{\Delta^{\mathrm{Y}}}
\newcommand{\gdiam}{\delta^{\mathrm{\cD}}}
\newcommand{\Gdiam}{\Delta^{\mathrm{\cD}}}
\newcommand{\gL}{\delta^{\mathrm{L}}}
\newcommand{\GL}{\Delta^{\mathrm{L}}}
\newcommand{\gLMS}{\delta^{\mathrm{L+M+S}}}
\newcommand{\GLMS}{\Delta^{\mathrm{L+M+S}}}

\begin{document}

\title{Topological implications of negative curvature for biological and social networks}

\author{R\'{e}ka Albert}
    \email{ralbert@phys.psu.edu}
    \homepage{www.phys.psu.edu/~ralbert} 
    \affiliation{Department of Physics, Pennsylvania State University, University Park, PA 16802}
\author{Bhaskar DasGupta}
    \email{dasgupta@cs.uic.edu}
    \homepage{www.cs.uic.edu/~dasgupta}
    \thanks{Author to whom correspondence should be sent.}
    \affiliation{Department of Computer Science, University of Illinois at Chicago, Chicago, IL 60607}
\author{Nasim Mobasheri} 
    \email{nmobas2@uic.edu}
    \affiliation{Department of Computer Science, University of Illinois at Chicago, Chicago, IL 60607}

\date{\today}



\begin{abstract}
Network measures that reflect the most salient properties of complex large-scale networks are in high demand in the network research community. 
In this paper we adapt a combinatorial measure of negative curvature (also called hyperbolicity) to parameterized finite networks, and show that 
a variety of biological and social networks {\em are} hyperbolic. This hyperbolicity property has strong implications on the higher-order connectivity and other 
topological properties of these networks. Specifically, we derive and prove bounds on the distance among shortest or approximately shortest paths in hyperbolic networks. 
We describe two implications of these bounds to cross-talk in biological networks, and to the existence of central, influential neighborhoods 
in both biological and social networks.
\end{abstract}

\maketitle

\section{Introduction}
For a large variety of complex systems, ranging from the World-Wide Web to metabolic networks, representation as a parameterized network 
and graph theoretical analysis of this network have led to many useful insights~\cite{Newman-book,AB-review}. In addition to 
established network measures such as the average degree, clustering coefficient or diameter, complex network researchers have proposed and evaluated a number 
novel network measures~\cite{rich-club,LM07,albert-et-al-2011,bassett-et-al-2011}. 
In this article we consider a combinatorial measure of negative curvature (also called hyperbolicity) of  parameterized finite networks and the implications of 
negative curvature on the higher-order connectivity and 
topological properties of these networks.

There are many ways in which the (positive or negative) curvature of a continuous surface or other similar spaces can be defined 
depending on whether the measure is to reflect the local or global properties of the underlying space. The specific notion of negative curvature that we 
use is an adoption of the hyperbolicity measure for a infinite metric space 
with bounded local geometry as originally proposed by Gromov~\cite{G87} using a so-called ``$4$-point condition''. 
We adopt this measure for parameterized finite discrete metric spaces induced by a network via all-pairs shortest paths and 
apply it to biological and social networks.
Recently, there has been a surge of empirical works measuring and analyzing the hyperbolicity of networks defined in this manner, 
and many real-world networks were observed to be hyperbolic 
in this sense.
For example, 
preferential attachment networks were shown to be scaled hyperbolic in~\cite{JL04,JLB07}, 
networks of high power transceivers in a wireless sensor network were empirically observed to have a tendency to be hyperbolic in~\cite{ALJKZ08}, 
communication networks at the IP layer and at other levels were empirically observed to be hyperbolic in~\cite{NS11,PKBV10}, 
extreme congestion at a very limited number of nodes in a very large traffic network was shown in~\cite{JLBB11}
to be caused due to hyperbolicity of the network together with minimum length routing, and 
the authors in~\cite{BPK10} showed how to efficiently map the topology of the Internet to a hyperbolic space.

Gromov's hyperbolicity measure adopted on a shortest-path metric of networks can also be visualized as a measure of 
the ``closeness'' of the original network topology to a tree topology~\cite{MSV11}.
Another popular measure used in both the bioinformatics and theoretical computer science literature 
is the {\em treewidth} measure first introduced by Robertson and Seymour~\cite{RS83}.
Many $\NP$-hard problems on general networks admit efficient polynomial-time solutions if restricted to classes of networks 
with bounded treewidth~\cite{B88}, just as several routing-related problems or the diameter estimation problem become easier if
the network has small hyperbolicity~\cite{CE07,CDEHV08,CDEHVX12,GL05}. 
However, as observed in~\cite{MSV11}, the two measures are quite different in nature: 
``the treewidth is more related to the least number of nodes whose removal changes the connectivity of the graph in a significant manner whereas the hyperbolicity
measure is related to comparing the geodesics of the given network with that of a tree''.
Other related research works on hyperbolic networks include estimating the distortion necessary to map hyperbolic metrics to tree metrics~\cite{ABKMRT07}
and studying the algorithmic aspects of several combinatorial problems on points in a hyperbolic space~\cite{KL06}.

\section{Hyperbolicity-related Definitions and Measures}

Let $G=(V,E)$ be a {\em connected} undirected graph of $n\geq 4$ nodes. 
We will use the following notations: 
\begin{itemize}
\item
$u \!\stackrel{\cP}{\leftrightsquigarrow}\! v$ denotes a path $\cP\equiv\left(u=u_0,u_1,\dots,u_{k-1},u_k=v\right)$ from node $u$ to node $v$ and 
$\ell(\cP)$ denotes the {\em length} (number of edges) of such a path.

\item
$u_i \!\stackrel{\cP}{\leftrightsquigarrow}\! u_j$ denotes the sub-path $\left(u_i,u_{i+1},\dots,u_j\right)$ of $\cP$ from $u_i$ to $u_j$.

\item
$u\!\stackrel{\mathfrak{s}}{\leftrightsquigarrow}\! v$ denotes a shortest path from node $u$ to node $v$ of length 
$d_{u,v}=\ell\big(u\!\stackrel{\mathfrak{s}}{\leftrightsquigarrow}\! v\big)$.
\end{itemize}
We introduce the hyperbolicity measures via the $4$-node condition as originally proposed by Gromov. 
Consider a quadruple of distinct nodes\footnote{If two or more nodes among $u_1,u_2,u_3,u_4$ are identical, then $\gadd_{u_1,u_2,u_3,u_4}=0$ due to the metric's triangle
inequality; thus it suffices to assume that the four nodes are distinct.} 
$u_1,u_2,u_3,u_4$, and let $\pi=\left(\pi_1,\pi_2,\pi_3,\pi_4\right)$ be a permutation of $\{1,2,3,4\}$ denoting 
a rearrangement of the indices of nodes such that 
\begin{multline*}
S_{u_1,u_2,u_3,u_4}=d_{u_{\pi_1},u_{\pi_2}}+d_{u_{\pi_3},u_{\pi_4}} 
\\
\leq M_{u_1,u_2,u_3,u_4}=d_{u_{\pi_1},u_{\pi_3}}+d_{u_{\pi_2},u_{\pi_4}} 
\\
\leq L_{u_1,u_2,u_3,u_4}=d_{u_{\pi_1},u_{\pi_4}}+d_{u_{\pi_2},u_{\pi_3}}
\end{multline*}
and let 
{$\gadd_{u_1,u_2,u_3,u_4} = \frac {L_{u_1,u_2,u_3,u_4}-M_{u_1,u_2,u_3,u_4}}{2}$.}
Considering all combinations of four nodes in a graph one can define a worst-case hyperbolicity\cite{G87} as 
\[
\gadd_{\mathrm{worst}}(G) = \max_{u_1,u_2,u_3,u_4} \left\{ \,\gadd_{u_1,u_2,u_3,u_4}\right\}
\]
and an average hyperbolicity as 
\[
\gadd_{\mathrm{ave}}(G) = \frac{1}{\binom{n}{4}} \hspace*{-0.2in} \sum_{\hspace*{0.2in}u_1,u_2,u_3,u_4} \hspace*{-0.2in} \gadd_{u_1,u_2,u_3,u_4}
\]
Note that $\gadd_{\mathrm{ave}}(G)$ is the expected value of $\gadd_{u_1,u_2,u_3,u_4}$ when the four nodes $u_1,u_2,u_3,u_4$ are picked
independently and uniformly at random from the set of all nodes. Both 
$\gadd_{\mathrm{worst}}(G)$ and $\gadd_{\mathrm{ave}}(G)$
can be trivially computed in O$\big(n^4\big)$ time for any graph $G$.

A graph $G$ is called $\delta$-hyperbolic if $\gadd_{\mathrm{worst}}(G) \leq\delta$. 
If $\delta$ is a small constant independent of the parameters of the graph, a $\delta$-hyperbolic 
graph is simply called a hyperbolic graph.
It is easy to see that if $G$ is a tree then 
$\gadd_{\mathrm{worst}}(G) = \gadd_{\mathrm{ave}}(G) = 0$. Thus all trees are hyperbolic graphs.

The hyperbolicity measure $\gadd_{\mathrm{worst}}$ considered in this paper for a metric space was originally used by Gromov in the 
context of group theory~\cite{G87} by observing that
many results concerning the fundamental group of a Riemann surface hold true in a more general context.
$\gadd_{\mathrm{worst}}$ is trivially infinite in the standard (unbounded) Euclidean space. Intuitively, 
a metric space has a finite value of $\gadd_{\mathrm{worst}}$ if it behaves metrically in the large scale as a negatively curved Riemannian manifold, and thus 
the value of $\gadd_{\mathrm{worst}}$ can be related to the standard scalar curvature of a Hyperbolic manifold. 
For example, a simply connected complete Riemannian manifold whose sectional curvature is below $\alpha<0$ has a value of 
$\gadd_{\mathrm{worst}}$ that is $\mathrm{O}\left({\left( \sqrt{-\alpha} \, \right)}^{-1}\right)$ (see~\cite{R96}).

In this paper we first show that a variety of biological and social networks are hyperbolic. 
We formulate and prove bounds on the existence of path-chords and on the distance among
shortest or approximately shortest paths in hyperbolic networks. We determine the implications of these bounds on {\em regulatory} networks, \IE, directed networks
whose edges correspond to regulation or influence. This category includes all the biological networks that we study in 
this paper. We also discuss the implications of our results on the 
region of influence of nodes in social networks. Some of the proofs of our theoretical results are adaptation of corresponding arguments in the
continuous hyperbolic space. All the proofs are presented in the appendix for the sake of completeness.

\renewcommand{\tabcolsep}{2pt}
\renewcommand{\arraystretch}{1}
\begin{table*}
\hspace*{-0.0in}
\begin{minipage}[c]{8.4cm}
\renewcommand{\arraystretch}{1.2}
\caption{\label{L3}Hyperbolicity and diameter values for biological networks.}
\scalebox{0.8}[0.8]
{
\begin{tabular}{l c c c c c c}
\\
\hline
\multicolumn{1}{c}{
\small
Network 
{\em id}
}
&
\small
reference
&
\begin{tabular}{c}
\small
Average 
\\
[-0.03in]
\small
degree 
\end{tabular}
&
\small
$\gadd_{\mathrm{ave}}(G)$
&
\small
$\gadd_{\mathrm{worst}}(G)$
&
\small
$\cD$
&
\small
$\dfrac{\gadd_{\mathrm{worst}}(G)}{\nicefrac{\cD}{2}}$
\\
[0.05in]
\hline
\small
{\bf 1}.  {\em E. coli} 
transcriptional & \small \cite{net1} 
&
\small $1.45$  	        &	\small $0.132$ &	\small $2$  	        &	\small $10$ & \small $0.400$
\\
\hline
\small {\bf 2}. Mammalian 
Signaling & \small \cite{net2} 
&
\small $2.04$	         	& \small $0.013$	& \small $3$	         	&  \small $11$ & \small $0.545$
\\
\hline
\small {\bf 3}. {\em E. Coli}
transcriptional & \small $\pmb{\sharp}$	
&
\small $1.30$	         	& \small $0.043$	& \small $2$	         	&  \small $13$ & \small $0.308$
\\
\hline
\small {\bf 4}. T LGL 
signaling & \small \cite{net6}
&
\small $2.32$	         	& \small $0.297$	& \small $2$	         	&   \small $7$ &  \small $0.571$
\\
\hline
\small {\bf 5}. 
{\em S. cerevisiae} 
transcriptional & \small \cite{net8}	
&
\small $1.56$	         	& \small $0.004$	& \small $3$	         	&  \small $15$ & \small $0.400$
\\
\hline
\small {\bf 6}. 
{\em C. elegans} 
Metabolic & \small \cite{net10-1} 
&
\small $4.50$	         	& \small $0.010$	& \small $1.5$	         	&  \small $7$ & \small $0.429$
\\
\hline
\small {\bf 7}. 
{\em Drosophila} 
segment polarity & \small \cite{Odell:2000} 
&
\small $1.69$	         	& \small $0.676$	& \small $4$	         	&   \small $9$ &  \small $0.889$
\\
\hline
\small {\bf 8}. 
ABA
signaling & \small \cite{LAA06}
&
\small $1.60$	         	& \small $0.302$	& \small $2$	         	&  \small $7$ & \small $0.571$
\\
\hline
\small {\bf 9}. 
Immune Response 	
Network & \small \cite{Thakar07}	
&
\small $2.33$	         	& \small $0.286$	& \small $1.5$	         	&  \small $4$ & \small $0.750$
\\
\hline
\small {\bf 10}. 
T Cell Receptor	
Signalling & \small \cite{julio07}
&
\small $1.46$	         	& \small $0.323$	& \small $3$	         	&  \small $13$ & \small $0.462$
\\
\hline
\small {\bf 11}. 
Oriented 
yeast PPI & \small \cite{net11}
&
\small $3.11$	         	& \small $0.001$	& \small $2$	         	& \small $6$ & \small $0.667$
\\
\hline
\multicolumn{7}{l}{
\small 
$^{\textstyle\pmb{\sharp}}\,$\cite[\small updated version]{net1}
}
\\
\multicolumn{7}{l}{
\small
$\,\,\,$ see www.weizmann.ac.il/mcb/UriAlon/Papers/networkMotifs/coli1\_1Inter\_st.txt
}
\end{tabular}
}
\renewcommand{\arraystretch}{1}
\end{minipage}
%
%
%
\hspace*{0.25in}
\begin{minipage}[c]{8.6cm}
\renewcommand{\tabcolsep}{2pt}
\renewcommand{\arraystretch}{1.2}
\caption{\label{L4}Hyperbolicity and diameter values for social networks.}
\scalebox{0.8}[0.8]
{
\begin{tabular}{l c c c c c c}
\\
\hline
\multicolumn{1}{c}{
\small
Network 
{\em id}
}
&
\small
reference
&
\begin{tabular}{c}
\small
Average 
\\
[-0.03in]
\small
degree 
\end{tabular}
&
\small
$\gadd_{\mathrm{ave}}(G)$
&
\small
$\gadd_{\mathrm{worst}}(G)$
&
\small
$\cD$
&
\small
$\dfrac{\gadd_{\mathrm{worst}}(G)}{\nicefrac{\cD}{2}}$
\\
[0.05in]
\hline
\small {\bf 1}.  
Dolphins	
social 
network & \small \cite{LSBHSD03} 
&
\small $5.16$  	&	\small $0.262$ &	\small $2$
& 
\small $8$ & 
\small $0.750$
\\
\hline
\small {\bf 2}.  
American 
College Football &  \small \cite{GN02} 
&
\small $10.64$  	        &	\small $0.312$ &	\small $2$ 
& 
\small $5$ & 
\small $0.800$
\\
\hline
\small {\bf 3}.  
Zachary 
Karate Club & \small \cite{Z77}
&
\small $4.58$  	        &	\small $0.170$ &	\small $1$
& 
\small $5$ & 
\small $0.400$
\\
\hline
\small {\bf 4}.  
Books about 
US Politics & \small $\pmb{\ddagger}$
&
\small $8.41$  	        &	\small $0.247$ &	\small $2$
& 
\small $7$ & 
\small $0.571$
\\
\hline
\small {\bf 5}.  
Sawmill 
communication & \small \cite{MM97} 
&
\small $3.44$  	        &	\small $0.162$ &	\small $1$
& 
\small $8$ & 
\small $0.250$
\\
\hline
\small {\bf 6}.  
Jazz musician & \small \cite{GD03} 
&
\small $27.69$  	        &	\small $0.140$ &	\small $1.5$
& 
\small $6$ & 
\small $0.500$
\\
\hline
\small {\bf 7}.  
Visiting ties 
in San Juan & \small \cite{LMCL53} 
&
\small $3.84$  	        &	\small $0.422$ &	\small $3$
& 
\small $9$ & 
\small $0.667$
\\
\hline
\small {\bf 8}.  
World Soccer 
data, 
$1998$ &  \small $\pmb{\dagger}$	
&
\small $3.37$  	        &	\small $0.270$ &	\small $2.5$
& 
\small $12$ & 
\small $0.286$ 
\\
\hline
\small {\bf 9}.  
Les 
Miserable & \small \cite{Knu93}
&
\small $6.51$  	        &	\small $0.278$ &	\small $2$ 
& 
\small $14$ & 
\small $0.417$ 
\\
\hline
\multicolumn{7}{l}{
\small $\,^{\textstyle\pmb{\ddagger}}\,$V. Krebs, \url{www.orgnet.com}, 
}
\\
\multicolumn{7}{l}{
\begin{tabular}{l}
\small $^{\textstyle\pmb{\dagger}}\,$Dagstuhl seminar: {\em Link Analysis and Visualization}, Dagstuhl 1-6, 2001;
\\
\small 
$\,\,\,$ \url{vlado.fmf.uni-lj.si/pub/networks/data/sport/football.htm}
\end{tabular}
}
\end{tabular}
}
\renewcommand{\tabcolsep}{6pt}
\renewcommand{\arraystretch}{1}
\end{minipage}
\end{table*}

\section{Results and Discussion}

Subsection~A examines in detail the hyperbolicity of an assorted list of diverse
biological and social networks. The remaining subsections of this section, namely subsections~B--E, 
state our findings on the implications of hyperbolicity of a network on various topological properties of the network.
For subsections~D, E, we first state our findings as applicable for biological or social networks, 
followed by a summary of formal mathematical results that led to such findings.
Because the precise bounds on topological features of a network as a function of hyperbolicty measures are quite mathematically
involved, we discuss these bounds in a somewhat simplified form
in subsections~B--E, leaving the precise bounds as theorems and proofs in the appendix.

\subsection{Hyperbolicity of Real Networks}
\label{realnet-sec}

We analyzed twenty well-known biological and social networks. The $11$ biological networks shown in Table~\ref{L3}
include $3$ transcriptional regulatory, 
$5$ signalling, $1$ metabolic, $1$ immune response and $1$ oriented protein-protein interaction networks.
Similarly, the $9$ social networks shown in Table~\ref{L4}
range from interactions in dolphin communities to the social network of 
jazz musicians. 
The hyperbolicity of the biological and directed social networks was computed by ignoring the direction of edges.
The hyperbolicity values were calculated by writing codes in C using standard algorithmic procedures.

As shown on Table~\ref{L3} and Table~\ref{L4}, the hyperbolicity values of almost all networks are small. 
If $\cD=\max_{u,v} \big\{ d_{u,v} \big\}$ is the diameter of the graph, then it is easy to see that  
$\gadd_{\mathrm{worst}}(G) \leq \nicefrac{\cD}{2}$, and thus small diameter indeed implies a small value of worst-case hyperbolicity. 
As can be seen on Table~\ref{L3} and Table~\ref{L4}, 
$\gadd_{\mathrm{worst}}(G)$ varies with respect to its worst-case bound of $\nicefrac{\cD}{2}$ from $25\%$ of $\nicefrac{\cD}{2}$ to no more than
$89\%$ of $\nicefrac{\cD}{2}$, and there does not seem to be a systematic dependence of 
$\gadd_{\mathrm{worst}}(G)$ 
on the number of nodes (which ranges from $18$ to $786$),  edges (from $42$ to $2742$),  or on the value of the diameter $\cD$.

For all the networks 
$\gadd_{\mathrm{ave}}(G)$ is one or two orders of magnitude smaller than $\gadd_{\mathrm{worst}}(G)$. 
Intuitively, this suggests that 
the value of $\gadd_{\mathrm{worst}}(G)$ may be a rare deviation
from typical values of $\gadd_{u_1,u_2,u_3,u_4}$ that one would obtain for most combinations of nodes 
$\left\{u_1,u_2,u_3,u_4\right\}$.

We additionally performed the following rigorous tests for hyperbolicity of our networks.

\subsubsection{Checking hyperbolicity via the scaled hyperbolicity approach}

An approach for testing hyperbolicity for finite graphs was introduced and used via ``scaled'' Gromov hyperbolicity 
in \cite{JLB07,NS11} for hyperbolicity defined via thin triangles and in~\cite{JLA11} for 
for hyperbolicity defined via the four-point condition as used in this paper. 
The basic idea is to ``scale'' the values of
$\gadd_{u_1,u_2,u_3,u_4}$ by a suitable scaling factor, say $\mu_{u_1,u_2,u_3,u_4}$, such that 
there exists a constant $0<\eps<1$ with the following property: 
\begin{itemize}
\item
the maximum achievable value of 
$\frac{\gadd_{u_1,u_2,u_3,u_4}}{\mu_{u_1,u_2,u_3,u_4}}$
is $\eps$ in 
the standard hyperbolic space or in the Euclidean space, and 

\item
$\frac{\gadd_{u_1,u_2,u_3,u_4}}{\mu_{u_1,u_2,u_3,u_4}}$
goes beyond $\eps$ in positively curved spaces.
\end{itemize}
We use the notation 
$\displaystyle\cD_{u_1,u_2,u_3,u_4}=\max_{i,j\in\{1,2,3,4\}} \left\{ d_{u_i,u_j} \right\}$ to indicate
the diameter of the subset of four nodes $u_1,u_2,u_3$ and $u_4$.
By using theoretical or empirical calculations, the authors in~\cite{JLA11} provide the bounds shown in Table~\ref{tabl1}.

\begin{table}
\caption{\label{tabl1}~\cite{JLA11} Various scaled Gromov hyperbolicities.}
\scalebox{0.9}[0.9]
{
\begin{tabular}{l l c c l}
\\
\hline
\multicolumn{1}{c}{\small Name} & \multicolumn{1}{c}{\small Notation} & \small $\mu_{u_1,u_2,u_3,u_4}$ & \small $\eps$ & \begin{tabular}{c} \small Method for \\ \small determining $\eps$ \end{tabular}
\\
[0.05in]
\hline
\small \begin{tabular}{l} diameter-scaled \\ hyperbolicity \end{tabular} & \small $\gdiam$ & \small $\cD_{u_1,u_2,u_3,u_4}$ & \small $0.2929$ & \small\hspace*{0.1in} empirical
\\
\hline
\small \begin{tabular}{l} $L$-scaled \\ hyperbolicity \end{tabular} &  \small $\gL$ & \small $L_{u_1,u_2,u_3,u_4}$ & \begin{tabular}{r l} \small & \small $\!\!\!\!\!\! \frac{\sqrt{2}-1}{2\sqrt{2}}$ \\ \small $\approx$ & \small $\!\! 0.1464$ \end{tabular} & \small \hspace*{0.1in}mathematical
\\
\hline
\small \begin{tabular}{l} $(L+M+S)$-scaled \\ hyperbolicity \end{tabular} &  \small $\gLMS$ & \begin{tabular}{l} \small $L_{u_1,u_2,u_3,u_4}$ \\ \small $\,\,+\,M_{u_1,u_2,u_3,u_4}$ \\ \small $\,\,\,\,\,\,+\,S_{u_1,u_2,u_3,u_4}$ \end{tabular} & \small $0.0607$ & \small \hspace*{0.1in}mathematical
\\
[0.2in]
\hline
\end{tabular}
}
\end{table}

\renewcommand{\tabcolsep}{5pt}
\renewcommand{\arraystretch}{1}
\begin{table*}
\begin{minipage}[c]{8.3cm}
\renewcommand{\arraystretch}{1.2}
\caption{\label{newtab1}$\Gscaled(G)$ values for biological networks for $Y\in \left\{ \cD,\, L,\, L+M+S\right\}$.}
\vspace*{-0.2in}
\begin{center}
\scalebox{1}[1]
{
\begin{tabular}{l c c c}
\\
\hline
\multicolumn{1}{c}{
\small
Network 
{\em id}
}
&
\small
$\Gdiam(G)$
&
\small
$\GL(G)$
&
\small
$\GLMS(G)$
\\
[0.05in]
\hline
\small
{\bf 1}.  {\em E. coli} 
transcriptional &  \small $0.0014$ &	\small $0.0018$
& \small $0.0015$
\\
\hline
\small {\bf 2}. Mammalian 
Signaling & \small $0.0021$	& \small $0.0018$
& \small $0.0022$
\\
\hline
\small {\bf 3}. {\em E. Coli}
transcriptional & \small $0.0006$	         	& \small $0.0006$
& \small $0.0007$
\\
\hline
\small {\bf 4}. T LGL 
signaling & \small $0.0228$	         	& 
\small $0.0221$
& \small $0.0318$
\\
\hline
\small {\bf 5}. 
{\em S. cerevisiae} 
transcriptional & \small $0.0031$	         	& 
\small $0.0032$
& \small $0.0033$
\\
\hline
\small {\bf 6}. 
{\em C. elegans} 
Metabolic & \small $0.0020$	         	& 
\small $0.0018$
& \small $0.0019$
\\
\hline
\small {\bf 7}. 
{\em Drosophila} 
segment polarity & \small $0.0374$	         	& 
\small $\mathbf{0.0558}$
& \small $0.0750$
\\
\hline
\small {\bf 8}. 
ABA
signaling & \small $0.0343$	         	& 
\small $0.0285$
& \small $0.0425$
\\
\hline
\small {\bf 9}. 
Immune Response 	
Network & \small $\mathbf{0.0461}$	         	& 
\small $0.0552$
& \small $\mathbf{0.0781}$
\\
\hline
\small {\bf 10}. 
T Cell Receptor	
Signalling & \small $0.0034$	         	& 
\small $0.0045$
& \small $0.0056$
\\
\hline
\small {\bf 11}. 
Oriented 
yeast PPI & \small $0.0013$	         	& 
\small $0.0009$
& \small $0.0012$
\\
\hline
\multicolumn{1}{r}{\small maximum} &  \small $0.0461$ & \small $0.0558$ & \small $0.0781$
\\
\hline
\end{tabular}
}
\end{center}
\renewcommand{\arraystretch}{1}
\end{minipage}
%
%
%
\hspace*{0.3in}
\begin{minipage}[c]{8.3cm}
\renewcommand{\tabcolsep}{5pt}
\renewcommand{\arraystretch}{1.2}
\caption{\label{newtab2}$\Gscaled(G)$ values for social networks for $Y\in \left\{ \cD,\, L,\, L+M+S\right\}$.}
\vspace*{-0.2in}
\begin{center}
\scalebox{1}[1]
{
\begin{tabular}{l c c c}
\\
\hline
\multicolumn{1}{c}{
\small
Network 
{\em id}
}
&
\small
$\Gdiam(G)$
&
\small
$\GL(G)$
&
\small
$\GLMS(G)$
\\
[0.05in]
\hline
\small {\bf 1}.  
Dolphins	
social 
network & \small $0.0115$ & 
\small $0.0120$ & 
\small $0.0168$
\\
\hline
\small {\bf 2}.  
American 
College Football &  \small $\mathbf{0.0435}$ & 
\small $\mathbf{0.0395}$ & 
\small $\mathbf{0.0577}$
\\
\hline
\small {\bf 3}.  
Zachary 
Karate Club &	\small $0.0195$ & 
\small $0.0249$ & 
\small $0.0284$
\\
\hline
\small {\bf 4}.  
Books about 
US Politics &	\small $0.0106$ & 
\small $0.0074$ & 
\small $0.0116$
\\
\hline
\small {\bf 5}.  
Sawmill 
communication &	\small $0.0069$ & 
\small $0.0068$ & 
\small $0.0085$
\\
\hline
\small {\bf 6}.  
Jazz musician &	\small $0.0097$ & 
\small $0.0117$ & 
\small $0.0124$
\\
\hline
\small {\bf 7}.  
Visiting ties 
in San Juan &	\small $0.0221$ & 
\small $0.0242$ & 
\small $0.0275$
\\
\hline
\small {\bf 8}.  
World Soccer 
data, 
$1998$ &	\small $0.0145$ & 
\small $0.0155$ & 
\small $0.0212$ 
\\
\hline
\small {\bf 9}.  
Les 
Miserable &	\small $0.0032$  & 
\small $0.0034$ & 
\small $0.0049$ 
\\
\hline
\multicolumn{1}{r}{\small maximum} &  \small $0.0435$ & \small $0.0395$ & \small $0.0577$
\\
\hline
\end{tabular}
}
\end{center}
\end{minipage}
\end{table*}
\renewcommand{\tabcolsep}{6pt}

\renewcommand{\arraystretch}{1}

We adapt the criterion proposed by Jonckheere, Lohsoonthorn and Ariaei~\cite{JLA11} to designate a given finite graph as hyperbolic 
by requiring a {\em significant} percentage of all possible subset of four nodes to
satisfy the $\eps$ bound. More formally, 
suppose that $G$ has $t$ connected components containing $n_1,n_2,\dots,n_t$ nodes, respectively ($\sum_{j=1}^t n_j=n$).
Let $0<\eta<1$ be a sufficiently high value indicating the confidence level in declaring the graph $G$ to be hyperbolic.
Then, we call our given graph $G$ to be (scaled) hyperbolic if and only if 
\begin{multline*}
\Gscaled(G)
=
\frac
{
\text{
\begin{tabular}{c}
\small
number of subset of four nodes $\left\{u_i,u_j,u_k,u_\ell\right\}$ 
\\
\small
such that $\gscaled_{u_i,u_j,u_k,u_\ell}>\eps$
\end{tabular}
}
}
{
\text{ 
\begin{tabular}{c}
\small
number of all possible combinations of four nodes 
\\
\small
that contribute to hyperbolicity
\end{tabular}
}
}
\\
=
\frac
{
\text{
\begin{tabular}{c}
\small
number of subset of four nodes $\left\{u_i,u_j,u_k,u_\ell\right\}$ 
\\
\small
such that $\gscaled_{u_i,u_j,u_k,u_\ell}>\eps$
\end{tabular}
}
}
{
\sum_{1\leq j\leq t \colon n_j>3} \binom{n_j}{4} 
}
<
1-\eta
\end{multline*}
The values of $\Gscaled(G)$ 
for our networks are shown in Table~\ref{newtab1} and Table~\ref{newtab2}.
It can be seen that, for all scaled hyperbolicity measures and for all networks, the value of $1-\eta$ is very close to zero.

\renewcommand{\tabcolsep}{2pt}
\renewcommand{\arraystretch}{1}
\begin{table*}
\renewcommand{\arraystretch}{1.2}
\caption{\label{newp1}$p$-values for the $\Gscaled(G)$ values for biological networks for $Y\in \left\{ \cD,\, L,\, L+M+S\right\}$.
In general, a $p$-value less than $0.05$ (shown in {\bf boldface}) is considered to be statistically significant,
and a $p$-value above $0.05$ is considered to be {\em not} statistically significant.
}
%
%
\scalebox{0.87}[0.87]
{
\begin{tabular}{c c c c c c c c c c c c c}
\\
\hline
& 
& 
\multicolumn{11}{c}{
Network 
{\em id}
}
\\
[0.05in]
\hline
& 
& 
\begin{tabular}{c}
\small
{\bf 1}.  
\\
[-0.05in]
\small
{\em E. coli} 
\end{tabular}
&
\begin{tabular}{c}
\small {\bf 2}. 
\\
[-0.05in]
\small Mammalian 
\\
[-0.05in]
\small Signaling 
\end{tabular}
&
\begin{tabular}{c}
\small {\bf 3}. 
\\
[-0.05in]
\small {\em E. Coli}
\\
[-0.05in]
\small
transcriptional 
\end{tabular}
&
\begin{tabular}{c}
\small {\bf 4}. 
\\
[-0.05in]
\small T LGL 
\\
[-0.05in]
\small
signaling 
\end{tabular}
&
\begin{tabular}{c}
\small {\bf 5}. 
\\
[-0.05in]
\small {\em S. cerevisiae} 
\\
[-0.05in]
\small
transcriptional 
\end{tabular}
&
\begin{tabular}{c}
\small {\bf 6}. 
\\
[-0.05in]
\small
{\em C. elegans} 
\\
[-0.05in]
\small
Metabolic 
\end{tabular}
&
\begin{tabular}{c}
\small {\bf 7}. 
\\
[-0.05in]
\small
{\em Drosophila} 
\\
[-0.05in]
\small
segment 
\\
[-0.05in]
\small
polarity 
\end{tabular}
&
\begin{tabular}{c}
\small {\bf 8}. 
\\
[-0.05in]
\small
ABA
\\
[-0.05in]
\small
signaling 
\end{tabular}
&
\begin{tabular}{c}
\small {\bf 9}. 
\\
[-0.05in]
\small
Immune 
\\
[-0.05in]
\small
Response 	
\\
[-0.05in]
\small
Network 
\end{tabular}
&
\begin{tabular}{c}
\small {\bf 10}. 
\\
[-0.05in]
\small
T Cell 
\\
[-0.05in]
\small
Receptor	
\\
[-0.05in]
\small
Signalling 
\end{tabular}
&
\begin{tabular}{c}
\small {\bf 11}. 
\\
[-0.05in]
\small
Oriented 
\\
[-0.05in]
\small
yeast 
\\
[-0.05in]
\small
PPI 
\end{tabular}
\\
\hline
\multirow{3}{*}{
\begin{tabular}{c}
\small 
$p$ 
\\
[-0.05in]
\small
values
\end{tabular}
} 
&
$\Gdiam$ 
& 
$\mathbf{0.0018}$
&
$\mathbf{<0.0001}$
&
$\mathbf{<0.0001}$
&
$\mathbf{<0.0001}$
&
\color{darkgray}
$0.3321$
&
$\mathbf{<0.0001}$
&
$\mathbf{<0.0001}$
&
$\mathbf{<0.0001}$
&
$\mathbf{<0.0001}$
&
$\mathbf{<0.0001}$
&
$\mathbf{<0.0001}$
\\
[0.05in]
\cline{2-13}
&
$\GL$ 
&
$\mathbf{<0.0001}$
&
$\mathbf{<0.0001}$
&
$\mathbf{<0.0001}$
&
\color{darkgray}
$0.011$
&
\color{darkgray}
$0.3434$
&
$\mathbf{<0.0001}$
&
$\mathbf{<0.0001}$
&
\color{darkgray}
$0.9145$
&
$\mathbf{<0.0001}$
&
$\mathbf{<0.0001}$
&
$\mathbf{<0.0001}$
\\
[0.05in]
\cline{2-13}
&
$\GLMS$ 
&
\color{darkgray}
$0.5226$
&
$\mathbf{<0.0001}$
&
$\mathbf{<0.0001}$
&
$\mathbf{<0.0001}$
&
\color{darkgray}
$0.3424$
&
$\mathbf{<0.0001}$
&
$\mathbf{<0.0001}$
&
\color{darkgray}
$0.3342$
&
$\mathbf{<0.0001}$
&
$\mathbf{<0.0001}$
&
$\mathbf{<0.0001}$
\\
[0.05in]
\hline
\end{tabular}
}
\end{table*}
\renewcommand{\arraystretch}{1}

\renewcommand{\tabcolsep}{2pt}
\renewcommand{\arraystretch}{1.2}
\begin{table*}
\caption{\label{newp2}$p$-values for the $\Gscaled(G)$ values for social networks for $Y\in \left\{ \cD,\, L,\, L+M+S\right\}$.
In general, a $p$-value less than $0.05$ (shown in {\bf boldface}) is considered to be statistically significant,
and a $p$-value above $0.05$ is considered to be {\em not} statistically significant.
}
\scalebox{1}[1]
{
\begin{tabular}{c c c c c c c c c c c}
\\
\hline
& 
&
\multicolumn{9}{c}{
Network 
{\em id}
}
\\
[0.05in]
\hline
& 
&
\begin{tabular}{c}
\small 
{\bf 1}.  
\\
[-0.05in]
\small 
Dolphins	
\\
[-0.05in]
social 
\\
[-0.05in]
\small 
network 
\end{tabular}
& 
\begin{tabular}{c}
\small 
{\bf 2}.  
\\
[-0.05in]
\small 
American 
\\
[-0.05in]
\small 
College 
\\
[-0.05in]
\small 
Football 
\end{tabular}
&
\begin{tabular}{c}
\small 
{\bf 3}.  
\\
[-0.05in]
\small 
Zachary 
\\
[-0.05in]
\small 
Karate 
\\
[-0.05in]
\small 
Club 
\end{tabular}
&
\begin{tabular}{c}
\small 
{\bf 4}.  
\\
[-0.05in]
\small 
Books 
\\
[-0.05in]
\small 
about 
\\
[-0.05in]
\small 
US Politics 
\end{tabular}
&
\begin{tabular}{c}
\small 
{\bf 5}.  
\\
[-0.05in]
\small 
Sawmill 
\\
[-0.05in]
\small 
communication 
\end{tabular}
&
\begin{tabular}{c}
\small 
{\bf 6}.  
\\
[-0.05in]
\small 
Jazz 
\\
[-0.05in]
\small 
musician 
\end{tabular}
&
\begin{tabular}{c}
\small 
{\bf 7}.  
\\
[-0.05in]
\small 
Visiting ties 
\\
[-0.05in]
\small 
in San Juan 
\end{tabular}
&
\begin{tabular}{c}
\small 
{\bf 8}.  
\\
[-0.05in]
\small 
World Soccer 
\\
[-0.05in]
\small 
data, 
$1998$ 
\end{tabular}
&
\begin{tabular}{c}
\small 
{\bf 9}.  
\\
[-0.05in]
\small 
Les 
\\
[-0.05in]
\small 
Miserable 
\end{tabular}
\\
\hline
\multirow{3}{*}{
\begin{tabular}{c}
\small 
$p$ 
\\
[-0.05in]
\small
values
\end{tabular}
} 
&
$\Gdiam$ 
& $\mathbf{<0.0001}$ & $\mathbf{<0.0001}$ & $\mathbf{<0.0001}$ & $\mathbf{<0.0001}$ & $\mathbf{<0.0001}$ & $\mathbf{<0.0001}$ & $\mathbf{<0.0001}$ & $\mathbf{<0.0001}$ & $\mathbf{<0.0001}$
\\
[0.05in]
\cline{2-11}
&
$\GL$ 
& $\mathbf{<0.0001}$ & $\mathbf{<0.0001}$ & $\mathbf{<0.0001}$ & $\mathbf{<0.0001}$ & $\mathbf{<0.0001}$ & $\mathbf{<0.0001}$ & \color{darkgray}$0.0779$ & $\mathbf{<0.0001}$ & $\mathbf{<0.0001}$
\\
[0.05in]
\cline{2-11}
&
$\GLMS$ 
& $\mathbf{<0.0001}$ & $\mathbf{<0.0001}$ & $\mathbf{<0.0001}$ & $\mathbf{<0.0001}$ & $\mathbf{<0.0001}$ & $\mathbf{<0.0001}$ & $\mathbf{<0.0001}$ & $\mathbf{<0.0001}$ & $\mathbf{<0.0001}$
\\
[0.05in]
\hline
\end{tabular}
}
\end{table*}
\renewcommand{\arraystretch}{1}

We next tested the statistical significance of the $\Gscaled(G)$ values by computing the statistical significance values (commonly 
called $p$-values) of these $\Gscaled(G)$ values for each network $G$ with respect to a null hypothesis model of the networks. 
We use a standard method used in the network science literature (\EG, see~\cite{albert-et-al-2011,net1}) for such purpose. 
For each network $G$, we generated $100$ randomized versions of the network  using a Markov-chain algorithm~\cite{KTV99} by swapping the endpoints of randomly selected pairs of
edges until $20\%$ of the edges was changed. We 
computed the values of
$\Gscaled\left(G_{\mathrm{rand}_1}\right)$, $\Gscaled\left(G_{\mathrm{rand}_2}\right)$, $\dots$, $\Gscaled\left(G_{\mathrm{rand}_{100}}\right)$.
We then used an (unpaired) one-sample student's {\sf t}-test to determine the probability that 
$\Gscaled(G)$ belongs to the same distribution as
$\Gscaled\left(G_{\mathrm{rand}_1}\right)$, $\Gscaled\left(G_{\mathrm{rand}_2}\right)$, $\dots$, $\Gscaled\left(G_{\mathrm{rand}_{100}}\right)$.

The $p$-values, tabulated in Table~\ref{newp1} and Table~\ref{newp2}, clearly show that all social networks and all except two biological networks 
can be classified as hyperbolic in a statistically significant manner, 
implying that the topologies of these networks are close to a ``tree topology''.
Indeed, for biological networks, the assumption of chain-like or tree-like topology
is frequently made in the traditional molecular biology literature~\cite{Alberts-1994}. 
Independent current observations also provide evidence of tree-like topologies for various biological networks, \EG, the average in/out degree of transcriptional 
regulatory networks~\cite{net1,Lee-et-al-2002} and of a mammalian signal transduction network~\cite{net2} is close to $1$, so cycles are very rare.


\begin{figure}
\epsfig{file=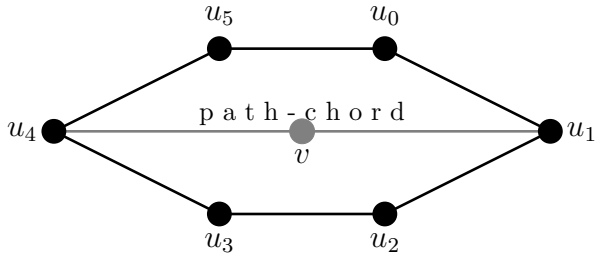}
\caption{\label{ch1}Path-chord of a cycle $\cC=\left(u_0.u_1,u_2,u_3,u_4,u_5,u_0\right)$.}
\end{figure}

\subsection{Hyperbolicity and crosstalk in regulatory networks}
\label{corrected-sec}

Let $\cC=\left(u_0,u_1,\dots,u_{k-1},u_0\right)$ be a cycle of $k\geq 4$ nodes. A {\em path-chord} of $\cC$ is defined to be a
path $u_i\!\stackrel{\cP}{\leftrightsquigarrow}\! u_j$ between two distinct nodes $u_i,u_j\in\cC$ such that the length of $\cP$ is 
less than $(i-j) \pmod{k}$ (see \FI{ch1}).
A path-chord of length $1$ is simply called a chord.

We find that large cycles without a path-chord imply large lower bounds on  hyperbolicity (see Theorem~\ref{kk1} in 
Section~\ref{kk1-sec} of the appendix). In particular, 
$G$ does not have a cycle of more than $4\,\gadd_{\mathrm{worst}}(G)$ nodes that does not have a path-chord.
Thus, for example, if $\gadd_{\mathrm{worst}}(G)<1$ then $G$ has no chordless cycle, \IE, $G$ is a chordal graph.
The intuition behind the proof of Theorem~\ref{kk1} is that if $G$ contains a long cycle without a 
path-chord then we can select four {\em almost} equidistant nodes on the cycle and these nodes give a large hyperbolicity value.
This general result has the following implications for regulatory networks:
\begin{itemize}
\item
If a node regulates itself through a long feedback loop (\EG, of length at least $6$ if 
$\gadd_{\mathrm{worst}}(G)=\nicefrac{3}{2}$) then this loop {\em must} have a path-chord. Thus it follows
that there exists a {\em shorter} feedback cycle through the same node. 

\item
A chord or short path-chord can be interpreted as {\em crosstalk} between two paths between a pair of nodes. 
With this interpretation, the following conclusion follows.
If one node in a regulatory network regulates another node through two sufficiently long paths, then there must be a crosstalk path between these
two paths. For example, assuming 
$\gadd_{\mathrm{worst}}(G)=\nicefrac{3}{2}$,  there must be a crosstalk path if the sum of lengths of the two paths is at least $6$. 
In general, the number of crosstalk paths between two paths increases {\em at least} linearly with the total length of the two paths. 
The general conclusion that can be 
drawn is that independent linear pathways that connect a signal to the same output node (\EG, transcription factor) are rare, and if multiple pathways exist then they are
interconnected through cross-talks.
\end{itemize}

\subsection{Shortest-path triangles and crosstalk paths in regulatory networks}
\label{sec1}

\noindent
{\bf ({\it a}) Result related to triplets of shortest paths} 
Originally, the hyperbolicity measure was introduced for infinite continuous metric spaces with negative curvature via the
concept of the ``thin'' and ``slim'' triangles (\EG, see~\cite{book}). For finite discrete metric spaces as induced by an undirected graph, 
one can analogously define a {\em shortest-path triangle} (or, simply a {\em triangle}) 
$\Delta_{\left\{u_0,u_1,u_2\right\}}$ as a set of three distinct nodes $u_0,u_1,u_2$ with a set of three shortest paths 
$\cP_{\Delta}\left(u_0,u_1\right)$, $\cP_{\Delta}\left(u_0,u_2\right)$, $\cP_{\Delta}\left(u_1,u_2\right)$
between $u_0$ and $u_1$, $u_0$ and $u_2$, and $u_1$ and $u_2$, respectively.
As illustrated on \FI{f4-old}, in hyperbolic networks we are guaranteed to find short 
paths\footnote{By a short path here, we mean a path whose length is at most a constant times 
$\gadd_{ \Delta_{ \{u_0,u_1,u_2 \} } }$ (note that $\gadd_{ \Delta_{ \{u_0,u_1,u_2 \} } } \leq \gadd_{\mathrm{worst}}(G)$).} between 
the nodes that make up $\cP_{\Delta}\left(u_0,u_1\right)$, $\cP_{\Delta}\left(u_0,u_2\right)$, $\cP_{\Delta}\left(u_1,u_2\right)$. This 
is formally stated in Theorem~\ref{lem1} in Section~\ref{lem1-sec} of the appendix. Moreover, as Corollary~\ref{lem1-cor} 
(in Section~\ref{lem1-sec} of the appendix)
states, we can have a small
Hausdorff distance between these shortest paths.
This result is a {\em proper} generalization of our previous result on path-cords. Indeed, in the special case when 
$u_1$ and $u_2$ are the same node the triangle becomes a shortest-path cycle involving the shortest paths between $u_0$ and $u_1$ and 
the short-cord result is obtained.

A proof of Theorem~\ref{lem1} is obtained by appropriate modification of a known similar bound 
for infinite continuous metric spaces.

The implications of this result for regulatory networks can be summarized as follows:

\begin{quote}
If we consider a feedback loop (cycle) or feed-forward loop formed by the shortest paths among three nodes, we can expect short cross-talk paths between 
these shortest paths. Consequently, the feedback or feed-forward loop will be {\em nested} with ``additional'' feed-back or feed-forward loops in which one of the paths will
be slightly longer. 
\end{quote}

\noindent
The above finding is empirically supported by the observation that network motifs (\EG, feed-forward or feed-back loops composed of three nodes and three edges) 
are often nested~\cite{albert-review}.

\begin{figure}
\epsfig{file=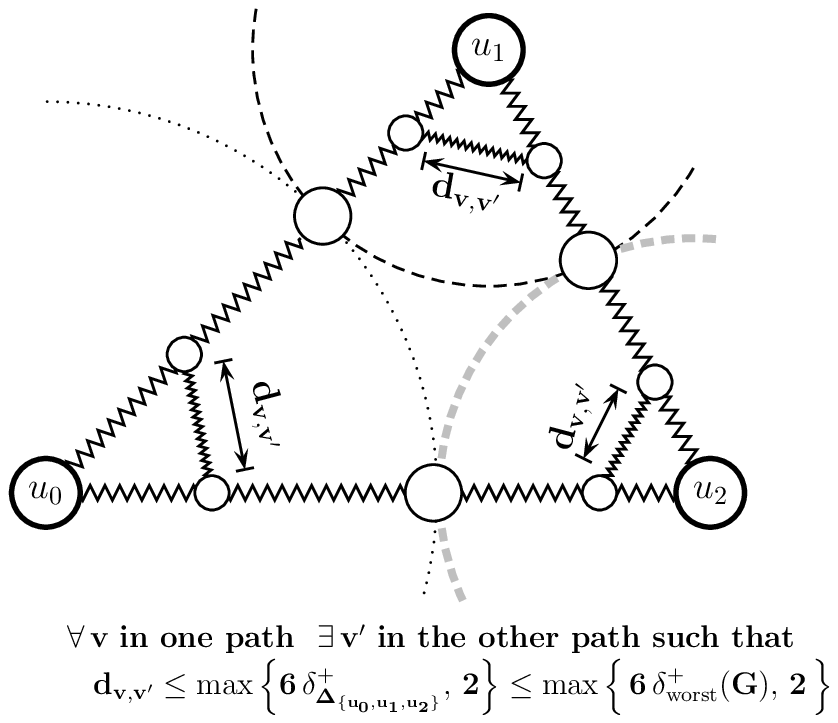}
\caption{\label{f4-old}An informal and simplified pictorial illustration of the claims in Section~\ref{sec1}(a).}
\end{figure}
%

\begin{figure}
\epsfig{file=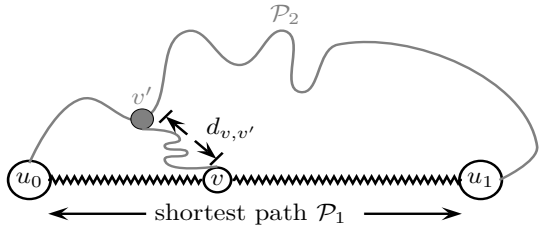}
\caption{\label{f8-old}An informal and simplified pictorial illustration of the claims in Section~\ref{sec1}(b).}
\end{figure}
%
%

\vspace*{0.1in}
\noindent
{\bf ({\it b}) Results related to the distance between two exact or approximate shortest paths between the same pair of 
nodes}
It is reasonable to assume that, when up- or down-regulation of a target node is mediated by two or more short 
paths~\footnote{Here by short paths we mean either a shortest path or an approximately shortest path whose length is not too much above the length of
a shortest path, \IE, a $\mu$-approximate short path or a $\eps$-additive-approximate short path,  
as defined in the subsequent ``Formal Justifications and Intuitions'' subsection, 
for small $\mu$ or small $\eps$, respectively.} 
starting from the {\em same} regulator node, additional
very long paths between the same regulator and target node do {\em not} contribute significantly to the target node's regulation. We refer to the short paths as relevant, and
to the long paths as irrelevant. Then, our finding can be summarized by saying that:

\begin{quote}
{\em almost all} relevant paths between two nodes have crosstalk paths between each other.
\end{quote}

\vspace*{0.1in}
\noindent
{\bf Formal Justifications and Intuitions} 
(see Theorem~\ref{lem2} and Corollary~\ref{cor1} in Section~\ref{lem2-sec} and 
Theorem~\ref{lem2-1} and Corollary~\ref{cor2-2} in Section~\ref{lem2-1-sec} of the appendix)

\vspace*{0.1in}
\noindent
We use the following two quantifications of ``approximately'' short paths:
\begin{itemize}
\item
A path $u_0\!\stackrel{\cP}{\leftrightsquigarrow}\! u_k =\big(u_0,u_1,\dots,u_k\big)$ is $\mu$-approximate short
provided $\ell\big(u_i\!\stackrel{\cP}{\leftrightsquigarrow}\! u_j\big)\leq \mu \, d_{u_i,u_j}$ for all $0\leq i<j\leq k$,

\item
A path $u_0\!\stackrel{\cP}{\leftrightsquigarrow}\! u_k$ is $\eps$-additive-approximate short provided
$\ell\left(\cP\right)\leq d_{u_0,u_k}+\eps$.
\end{itemize}
A mathematical justification for the claim then is provided by two separate theorems and their corollaries:
\begin{itemize}
\item
Let $\cP_1$ and $\cP_2$ be a shortest path and an arbitrary path, respectively, between two nodes $u_0$ and $u_1$.
Then, Theorem~\ref{lem2} and Corollary~\ref{cor1} implies that,  for every node $v$ on $\cP_1$, there exists a node $v'$ on $\cP_2$ such that 
$d_{v,v'}$ depends linearly on $\gadd_{\mathrm{worst}}(G)$, only logarithmically on the length of $\cP_2$ and 
does {\em not} depend on the size or any other parameter of the network. 

To obtain this type of bound, one needs to apply Theorem~\ref{lem1} on $u_0$, $u_1$ and the middle node of
the path $\cP_2$ and then use the same approach recursively on a part of the path $\cP_2$ containing at most $\left\lceil\nicefrac{\left(\cP_2\right)}{2}\right\rceil$ edges. 
The depth of the level of recursion provides the logarithmic 
factor in the bound.

\item
If $\cP_1$ and $\cP_2$ are two short paths between $u_0$ and $u_1$ then 
Theorem~\ref{lem2-1} and Corollary~\ref{cor2-2} imply that 
the Hausdorff distance between $\cP_1$ and $\cP_2$ 
depends on $\gadd_{\mathrm{worst}}(G)$ only and does {\em not} depend on the size or any other parameter of the network. 

Intuitively, Theorem~\ref{lem2-1} and Corollary~\ref{cor2-2} can be thought of as generalizing and improving the bound in Theorem~\ref{lem2} for approximately short paths. 
\end{itemize}
%

%
%
\begin{figure}
\epsfig{file=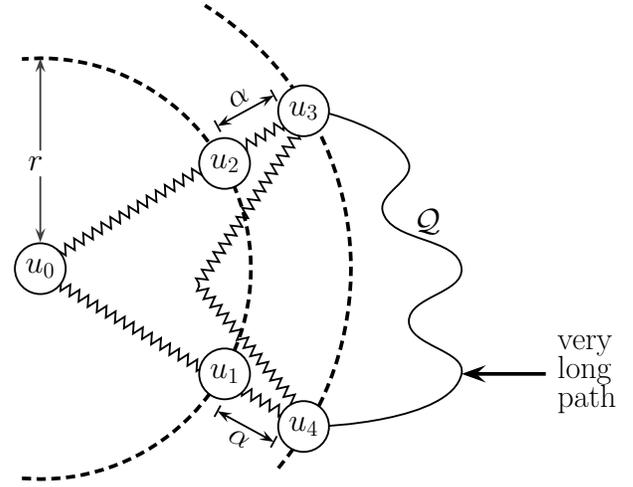}
\caption{\label{f11-old}An informal and simplified pictorial illustration of claim $(\star)$ in Section~\ref{essen-sec}. As the
nodes $u_3$ and $u_4$ move further away from the center node $u_0$, the shortest path between them bends more towards $u_0$
and any path between them that does not involve a node in the ball $\cup_{r'\leq r}B_{r'}\left(u_0\right)$ is long enough.}
\end{figure}
%

\subsection{Identifying essential edges in the regulation between two nodes}
\label{essen-sec}

For a given $\xi>0$ and a node $u$, let $\mathcal{B}_{\xi}\,(u)=\left\{ \, v \, | \, d_{u,v}=\xi \, \right\}$ denote the 
``boundary of the $\xi$-neighborhood'' of $u$, \IE, the set of all nodes at a distance of precisely $\xi$ from $u$. 
Our two findings in the present context are as stated in {\bf (I)} and {\bf (II)} below.

\vspace*{0.1in}
\noindent
{\bf (I) Identifying relevant paths between a source and a target node}
Suppose that we pick a node $v$ and consider the {\em strict} $\xi$-neighborhood 
of $v$
\[
\displaystyle N^+_{\xi}(v)=\bigcup_{r\leq \xi}\mathcal{B}_{r'}\left(v\right)\setminus \big\{ u \, | \, \text{degree of $u$ is one} \big\}
\]
(\IE, the set of all nodes, excluding nodes of degree $1$, that are at a distance at most $\xi$ from $u$)
for a sufficiently large $\xi$.
Consider two nodes $u_1$ and $u_2$ on the boundary of this neighborhood, \IE, at a distance $\xi$ from $v$. 
Then, the following holds:

\begin{quote}
($\star$)
the relevant (short) regulatory paths
between $u_1$ and $u_2$ do {\em not} leave the neighborhood, \IE, {\em all} the edges in the relevant regulatory paths are in the neighborhood. 
\end{quote}

\noindent
Thus, {\em only} the edges inside the neighborhood are relevant to the regulation among this pair of nodes.

This result can be adapted to find the most relevant paths between the input node
$u_{\mathrm{source}}$
and output node 
$u_{\mathrm{target}}$
of a signal transduction network. In many situations, for example
when the signal transduction network is inferred from undirected protein-protein interaction data, a large number of paths can potentially be included in the 
signal transduction network as the protein-protein interaction network has a large connected component with a small average path length~\cite{albert-review}. 
There is usually {\em no} prior knowledge on which of the existing paths are relevant to the signal transduction network. A hyperbolicity-based method is to first find a central node 
$u_{\mathrm{central}}$
which is at equal distance between 
$u_{\mathrm{source}}$
and 
$u_{\mathrm{target}}$, 
and is on the shortest, or close to shortest, path between 
$u_{\mathrm{source}}$
and 
$u_{\mathrm{target}}$.
Then one constructs the neighborhood around 
$u_{\mathrm{central}}$
such that 
$u_{\mathrm{source}}$
and 
$u_{\mathrm{target}}$
are on the boundary of this neighborhood. 
Applying this result, the paths relevant to the signal transduction network are inside the neighborhood, and the paths that go out of the neighborhood are irrelevant.
See \FI{f11-old} for a pictorial illustration of this implication.

%
%
\begin{figure}
\epsfig{file=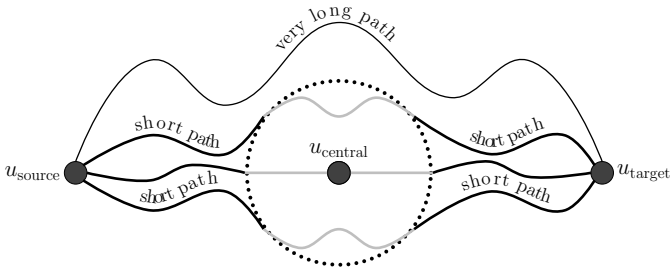}
\caption{\label{fig-rel}An informal and simplified pictorial illustration of claim $(\star\!\star)$ in Section~\ref{essen-sec}.
Knocking out the nodes in a small neighborhood of $u_{\mathrm{central}}$ 
cuts off all relevant (short) regulation between $u_{\mathrm{source}}$ and $u_{\mathrm{target}}$.}
\end{figure}
%

\vspace*{0.1in}
\noindent
{\bf (II) Finding essential nodes}
Again, consider an input node 
$u_{\mathrm{source}}$ and output node 
$u_{\mathrm{target}}$
of a signal transduction network, 
and let 
$u_{\mathrm{central}}$ be a central node 
which is on the shortest path between them and at approximately equal distance between 
$u_{\mathrm{source}}$
and 
$u_{\mathrm{target}}$.
Our results show that\footnote{O and $\Omega$ are the standard notations used 
in analyzing asymptotic upper and lower bounds in the computer science literature: given
two functions $f(n)$ and $g(n)$ of a variable $n$, $f(n)=\mathrm{O}(g(n))$ (respectively, $f(n)=\Omega(g(n))$ provided there exists two constants $n_0,c>0$ such that
$f(n)\leq c\,g(n)$ (respectively, $f(n)\geq c\,g(n)$) for $n\geq n_0$.}

\begin{quote}
($\star\!\star$)
if one constructs a small $\xi$-neighbourhood around $u_{\mathrm{central}}$ 
with $\xi=\text{O}\left(\gadd_{\mathrm{worst}}(G)\right)$,  
then {\em all} relevant (short or approximately short) paths 
between 
$u_{\mathrm{source}}$
and 
$u_{\mathrm{target}}$
must include a {\em node} in this $\xi$-neighborhood. Therefore, ``knocking out'' the nodes in this $\xi$-neighborhood cuts off
all relevant regulatory paths between 
$u_{\mathrm{source}}$
and 
$u_{\mathrm{target}}$. 
\end{quote}

\noindent
See \FI{fig-rel} for a pictorial illustration of this implication.
Note that the size $\xi$ of the neighborhood depends only on $\gadd_{\mathrm{worst}}(G)$ which, as our empirical results indicate, is usually
a small constant for real networks.

\vspace*{0.1in}
\noindent
{\bf Formal Justifications and Intuitions for $(\star)$ and $(\star\!\star)$} 
(see Theorem~\ref{lem3} and Corollary~\ref{cor2} in Section~\ref{lem3-sec} of the appendix)

\vspace*{0.1in}
\noindent
Suppose that we are given the following:
\begin{itemize}
\item
three integers $\kappa\geq 4$, $\alpha>0$, 
\\
$r>\left(\frac{\kappa}{2}-1\right)\left( 6\,\gadd_{\mathrm{worst}}(G) + 2 \right)$,

\item
five nodes $u_0,u_1,u_2,u_3,u_4$ such that
\begin{itemize}
\item
$u_1,u_2\in B_r\left(u_0\right)$ with $d_{u_1,u_2}\geq \frac{\kappa}{2} \, \left( 6\,\gadd_{\mathrm{worst}}(G) + 2 \right)$,

\item
$d_{u_1,u_4}=d_{u_2,u_3}=\alpha$. 
\end{itemize}
\end{itemize}
Then, $(\star)$ and $(\star\!\star)$ are implied by 
following type of {\em asymptotic} bounds provided by Theorem~\ref{lem3} and Corollary~\ref{cor2}:

\begin{quote}
For a suitable positive value $\lambda=\mathrm{O}\big(\gadd_{\mathrm{worst}}(G)\,\big)$, if $d_{u_1,u_4}=d_{u_2,u_3}=\alpha>\lambda$ then  
one of the following is true for any path $\cQ$ between $u_3$ and $u_4$ that does not involve a node in $\cup_{r'\leq r}\mathcal{B}_{r'}\left(u_0\right)$:
\begin{itemize}
\item
$\cQ$ does not exist (\IE, $\ell(\cQ)\geq n$), or 

\item
$\cQ$ is much longer than a shortest path between the two nodes, \IE, if 
$\cQ$ is a $\mu$-approximate short path or a $\eps$-additive-approximate short path then 
$\mu$ or $\eps$ is large.
\end{itemize}
A pessimistic estimate shows that a value of $\lambda$ that is about $6\,\gadd_{\mathrm{worst}}(G) + 2$ suffices.
As we subsequently observe, for real networks the bound is much better, about 
$\lambda\approx\gadd_{\mathrm{worst}}(G)$.
\end{quote}

\begin{table*}
\renewcommand{\tabcolsep}{3pt}
\renewcommand{\arraystretch}{1}
\caption{\label{KO-prev}Effect of the prescribed neighborhood in claim ($\star$) on all edges in relevant paths.}
\scalebox{1}[1]
{
\begin{tabular} { l   c   c   c   c c   c   c   c }
\\
[-0.05in]
\multicolumn{9}{c}{
                    \begin{tabular}{r c l}
                             \small $\mathcal{S}\mathcal{P}$ & \small : & \small shortest path between $u_{\mathrm{source}}$ and $u_{\mathrm{target}}$
                                \\
                             \small $\mathcal{S}\mathcal{P}^{+1}$& \small : & \small paths between $u_{\mathrm{source}}$ and $u_{\mathrm{target}}$ with one extra edge than $\mathcal{S}\mathcal{P}$ 
                              ($1$-additive-approximate short path)
                                \\
                             \small $\mathcal{S}\mathcal{P}^{+2}$& \small : & \small paths between $u_{\mathrm{source}}$ and $u_{\mathrm{target}}$ with two extra edges than $\mathcal{S}\mathcal{P}$
                              ($2$-additive-approximate short path)
                                \\
                             \small $N^+_{\xi} \left(u_{\,\mathrm{central}}\right)$& 
                                   \small : & \small strict $\xi=d_{u_{\,\mathrm{source}}\,,\,u_{\,\mathrm{target}}}$ neighborhood of $u_{\mathrm{central}}$ 
                                \\
                              \small $n$ & : & \small size (number of nodes) of the network
                                \\
                             \large $\nicefrac{N^+_{\xi} \left(u_{\,\mathrm{central}}\right)}{n}$& 
                                   \small : & \small fraction of strict $\xi=d_{u_{\,\mathrm{source}}\,,\,u_{\,\mathrm{target}}}$ neighborhood of $u_{\mathrm{central}}$ with respect to the size of the network
                    \end{tabular}
                   }
\\
[-0.1in]
\\
\hline
\\
[-0.1in]
\multicolumn{1}{c}{
\begin{tabular}{c}
\small
Network 
\\
\small
name 
\end{tabular}
}
& 
\large$u_{\mathrm{source}}$ &  \large$u_{\mathrm{target}}$ & $d_{u_{\,\mathrm{source}}\,,\,u_{\,\mathrm{target}}}$ & \large$u_{\mathrm{central}}$ & 
        $\dfrac{N^+_{\xi} \left(u_{\,\mathrm{central}}\right)}{n}$ 
       &
         \begin{tabular}{c} 
                                  \small \% of $\mathcal{S}\mathcal{P}$ 
                                      \\
                                   \small with every 
                                      \\
                                   \small edge in the 
                                     \\
                                   \small neighborhood 
                                      \\
                                   \small of
                                   \small claim ($\star$)
                               \end{tabular}
&
         {  \begin{tabular}{c} 
                                  \small \% of $\mathcal{S}\mathcal{P}^{+1}$
                                      \\
                                   \small with every 
                                      \\
                                   \small edge in the 
                                      \\
                                   \small neighborhood 
                                      \\
                                   \small of 
                                   \small claim ($\star$)
                               \end{tabular}
                           }
&
         {  \begin{tabular}{c} 
                                  \small \% of $\mathcal{S}\mathcal{P}^{+2}$
                                      \\
                                   \small with every
                                      \\
                                   \small edge in the 
                                      \\
                                   \small neighborhood 
                                      \\
                                   \small of
                                   \small claim ($\star$)
                               \end{tabular}
                           }
%
\\
\hline
\multirow{4}{*}{  \begin{tabular}{c} 
                       \small Network 1:  
                         \\
                       \small {\em E. coli} 
                         \\
                       \small transcriptional
                  \end{tabular}
               }
&
\multirow{2}{*}{ \small fliAZY  } 
&
\multirow{2}{*}{ \small arcA } 
&
\multirow{2}{*}{ \small $4$ } 
&
\small CaiF & \small 0.20 & \small 100\% & \small 100\% & \small 18\% 
\\
\cline{5-9}
& & & &
\small crp & \small 0.27 & \small 100\% & \small 100\% & \small 70\% 
\\
\cline{2-9}
&
\multirow{2}{*}{ \small fecA  } 
&
\multirow{2}{*}{ \small aspA } 
&
\multirow{2}{*}{ \small $6$ } 
&
\small crp & \small 0.43 & \small 100\% & \small 100\% & \small 100\% 
\\
\cline{5-9}
& & & & 
\small sodA & \small 0.28 & \small 100\% & \small 100\% & \small 62\% 
\\
\hline
\multirow{5}{*}{  \begin{tabular}{c} 
                       \small Network 4: 
                         \\
                        \small T-LGL 
                         \\
                        \small signaling
                  \end{tabular}
               }
&
{ \small IL15 } 
&
{ \small Apoptosis } 
&
\small $4$ 
&
\small GZMB & \small 0.37 & \small 100\% & \small 66\% & \small 40\% 
\\
\cline{2-9}
&
\multirow{3}{*}{ \small PDGF } 
&
\multirow{3}{*}{ \small Apoptosis } 
&
\multirow{3}{*}{ \small $6$ } 
&
\small IL2, NKFB & \small $\,\,\,$0.72,0.59 & \small 100\% & \small 100\% & \small 100\% 
\\
\cline{5-9}
& & & &
\small Ceramide & \small 0.60 & \small 80\% & \small 64\% & \small 36\% 
\\
\cline{5-9}
& & & & 
\small MCL1 & \small 0.59 & \small 80\% & \small 88\% & \small 93\% 
\\
\cline{2-9}
&
{ \small stimuli } 
&
{ \small Apoptosis } 
&
\small $4$
&
\small GZMB & \small 0.37 & \small 100\% & \small 100\% & \small 100\% 
\\
\hline
\end{tabular}
\renewcommand{\arraystretch}{1}
\renewcommand{\tabcolsep}{6pt}
}
\end{table*}

\vspace*{0.1in}
\noindent
{\bf Empirical evaluation of $(\star)$}

\vspace*{0.1in}
\noindent
We empirically investigated the claim in ($\star$) on relevant paths passing through a neighborhood of a central node for 
the following two biological networks:

\begin{quote}
\begin{description}
\item[Network 1:]  
{\em E. coli} transcriptional, and
\item[Network 4:] 
T-LGL signaling.
\end{description}
\end{quote}

\noindent
For each network we selected a few biologically relevant source-target pairs. 
For each such pair $u_{\mathrm{source}}$ and $u_{\mathrm{target}}$, 
we found the shortest path(s) between them. For each such shortest path, a central node $u_{\mathrm{central}}$ was identified. 
We then considered the $\xi$-neighborhood of $u_{\mathrm{central}}$ such that both 
both $u_{\mathrm{source}}$ and $u_{\mathrm{target}}$ are on the boundary of the neighborhood, and for each such neighborhood 
we determined what percentage of shortest or approximately short path (with one or two extra edges compared to shortest paths) 
between $u_{\mathrm{source}}$ and $u_{\mathrm{target}}$
had {\em all} edges in this neighborhood. The results, tabulated in Table~\ref{KO-prev}, support ($\star$).

\begin{table*}
\renewcommand{\tabcolsep}{3pt}
\renewcommand{\arraystretch}{1}
\caption{\label{KO}The effect of the size of the neighborhood in mediating short paths.}
\scalebox{0.97}[1]
{
\begin{tabular} { l   c   c   c   l   lr   lr   lr }
\multicolumn{11}{c}{
                    	\begin{tabular}{l}
                             \small $\mathcal{S}\mathcal{P}$ : shortest path between $u_{\mathrm{source}}$ and $u_{\mathrm{target}}$
                                \\
                             \small $\mathcal{S}\mathcal{P}^{+1}$: paths between $u_{\mathrm{source}}$ and $u_{\mathrm{target}}$ with one extra edge than $\mathcal{S}\mathcal{P}$
                              ($1$-additive-approximate short path)
                                \\
                             \small $\mathcal{S}\mathcal{P}^{+2}$: paths between $u_{\mathrm{source}}$ and $u_{\mathrm{target}}$ with two extra edges than $\mathcal{S}\mathcal{P}$
                              ($2$-additive-approximate short path)
                           \end{tabular}
                   }
\\
[-0.1in]
\\
\hline
\multicolumn{1}{c}{
\begin{tabular}{c}
\small Network 
\\
\small name 
\end{tabular}
}
& 
\large$u_{\mathrm{source}}$ &  \large$u_{\mathrm{target}}$ &  $d_{u_{\,\mathrm{source}}\,,\,u_{\,\mathrm{target}}}$ & \large$u_{\mathrm{central}}$ & 
         \multicolumn{2}{c}{  \begin{tabular}{c} 
                                  \small \% of $\mathcal{S}\mathcal{P}$ 
                                   \small with a node
                                      \\
                                   \small in $\xi$-neighborhood
                               \end{tabular}
                           }
&
         \multicolumn{2}{c}{  \begin{tabular}{c} 
                                  \small \% of $\mathcal{S}\mathcal{P}^{+1}$ 
                                   \small with a node
                                      \\
                                   \small in $\xi$-neighborhood
                               \end{tabular}
                           }
&
         \multicolumn{2}{c}{  \begin{tabular}{c} 
                                  \small \% of $\mathcal{S}\mathcal{P}^{+2}$ 
                                   \small with a node
                                      \\
                                   \small in $\xi$-neighborhood
                               \end{tabular}
                           }
%
\\
\hline
\multirow{4}{*}{  \begin{tabular}{c} 
                       \small Network 1:  
                         \\
                       \small {\em E. coli} 
                         \\
                       \small transcriptional
                         \\
                       \small $\gadd_{\mathrm{worst}}(G)=2$
                  \end{tabular}
               }
&
\multirow{2}{*}{ \small fliAZY  } 
&
\multirow{2}{*}{ \small arcA } 
&
\multirow{2}{*}{ \small $4$ } 
&
{ \small CaiF } 
          & \small $\,\,\,\,\,\,\,\,\xi=1$ & \small 100\%$\,\,\,\,\,\,\,\,\,\,$ & \small $\,\,\,\,\,\,\,\,\,\,\,\,\,\xi=1$ & \small $\!\!\!$71\%$\,\,\,\,\,\,\,\,\,\,$ & \small $\,\,\,\,\,\,\,\,\,\,\,\,\,\xi=1$ & \small 59\%$\,\,\,\,\,\,\,\,\,\,$ 
\\
\cline{5-11}
& & & & { \small crp } 
          & \small $\,\,\,\,\,\,\,\,\xi=1$ & \small 100\%$\,\,\,\,\,\,\,\,\,\,$ & \small $\,\,\,\,\,\,\,\,\,\,\,\,\,\xi=1$ & \small 100\%$\,\,\,\,\,\,\,\,\,\,$ & \small $\,\,\,\,\,\,\,\,\,\,\,\,\,\xi=1$ & \small 100\%$\,\,\,\,\,\,\,\,\,\,$ 
\\
\cline{2-11}
&
\multirow{2}{*}{ \small fecA  } 
&
\multirow{2}{*}{ \small aspA } 
&
\multirow{2}{*}{ \small $6$ } 
&
{ \small crp } 
          & \small $\,\,\,\,\,\,\,\,\xi=1$ & \small 100\%$\,\,\,\,\,\,\,\,\,\,$ & \small $\,\,\,\,\,\,\,\,\,\,\,\,\,\xi=1$ & \small 100\%$\,\,\,\,\,\,\,\,\,\,$ & \small $\,\,\,\,\,\,\,\,\,\,\,\,\,\xi=1$ & \small 100\%$\,\,\,\,\,\,\,\,\,\,$ 
\\
\cline{5-11}
& & &     &           { \small sodA } 
          & \small $\,\,\,\,\,\,\,\,\xi=1$ & \small 100\%$\,\,\,\,\,\,\,\,\,\,$ & \small $\,\,\,\,\,\,\,\,\,\,\,\,\,\xi=1$ & \small 100\%$\,\,\,\,\,\,\,\,\,\,$ & \small $\,\,\,\,\,\,\,\,\,\,\,\,\,\xi=1$ & \small 100\%$\,\,\,\,\,\,\,\,\,\,$ 
\\
\hline
\multirow{10}{*}{  \begin{tabular}{c} 
                       \small Network 4: 
                         \\
                        \small T-LGL 
                         \\
                        \small signaling
                         \\
                       \small $\gadd_{\mathrm{worst}}(G)=2$
                  \end{tabular}
               }
&
{ \small IL15 } 
&
{ \small apoptosis } 
&
\small $4$
&
{ \small GZMB } 
          & \small $\,\,\,\,\,\,\,\,\xi=1$ & \small 100\%$\,\,\,\,\,\,\,\,\,\,$ & \small $\,\,\,\,\,\,\,\,\,\,\,\,\,\xi=1$ & \small 100\%$\,\,\,\,\,\,\,\,\,\,$ & \small $\,\,\,\,\,\,\,\,\,\,\,\,\,\xi=1$ & \small 100\%$\,\,\,\,\,\,\,\,\,\,$ 
\\
\cline{2-11}
&
\multirow{8}{*}{ \small PDGF } 
&
\multirow{8}{*}{ \small apoptosis } 
&
\multirow{8}{*}{ \small $6$ } 
&
\multirow{2}{*}{ \small IL2 } 
          & \small $\,\,\,\,\,\,\,\,\xi=1$ & \small 80\%$\,\,\,\,\,\,\,\,\,\,$ & \small $\,\,\,\,\,\,\,\,\,\,\,\,\,\xi=1$ & \small 82\%$\,\,\,\,\,\,\,\,\,\,$ & \small $\,\,\,\,\,\,\,\,\,\,\,\,\,\xi=1$ & \small 93\%$\,\,\,\,\,\,\,\,\,\,$ 
\\
& & & &
          & \small $\,\,\,\,\,\,\,\,\xi=2$ & \small 100\%$\,\,\,\,\,\,\,\,\,\,$ & \small $\,\,\,\,\,\,\,\,\,\,\,\,\,\xi=2$ & \small 100\%$\,\,\,\,\,\,\,\,\,\,$ & \small $\,\,\,\,\,\,\,\,\,\,\,\,\,\xi=2$ & \small 100\%$\,\,\,\,\,\,\,\,\,\,$ 
\\
\cline{5-11}
& & & &
\multirow{2}{*}{ \small NFKB } 
          & \small $\,\,\,\,\,\,\,\,\xi=1$ & \small 80\%$\,\,\,\,\,\,\,\,\,\,$ & \small $\,\,\,\,\,\,\,\,\,\,\,\,\,\xi=1$ & \small 86\%$\,\,\,\,\,\,\,\,\,\,$ & \small $\,\,\,\,\,\,\,\,\,\,\,\,\,\xi=1$ & \small 76\%$\,\,\,\,\,\,\,\,\,\,$ 
\\
& & & &
          & \small $\,\,\,\,\,\,\,\,\xi=2$ & \small 100\%$\,\,\,\,\,\,\,\,\,\,$ & \small $\,\,\,\,\,\,\,\,\,\,\,\,\,\xi=2$ & \small 100\%$\,\,\,\,\,\,\,\,\,\,$ & \small $\,\,\,\,\,\,\,\,\,\,\,\,\,\xi=2$ & \small 100\%$\,\,\,\,\,\,\,\,\,\,$ 
\\
\cline{5-11}
& & & & 
\multirow{2}{*}{ \small Ceramide } 
          & \small $\,\,\,\,\,\,\,\,\xi=1$ & \small 40\%$\,\,\,\,\,\,\,\,\,\,$ & \small $\,\,\,\,\,\,\,\,\,\,\,\,\,\xi=1$ & \small 23\%$\,\,\,\,\,\,\,\,\,\,$ & \small $\,\,\,\,\,\,\,\,\,\,\,\,\,\xi=1$ & \small 40\%$\,\,\,\,\,\,\,\,\,\,$ 
\\
& & & & 
          & \small $\,\,\,\,\,\,\,\,\xi=2$ & \small 100\%$\,\,\,\,\,\,\,\,\,\,$ & \small $\,\,\,\,\,\,\,\,\,\,\,\,\,\xi=2$ & \small 100\%$\,\,\,\,\,\,\,\,\,\,$ & \small $\,\,\,\,\,\,\,\,\,\,\,\,\,\xi=2$ & \small 100\%$\,\,\,\,\,\,\,\,\,\,$ 
\\
\cline{5-11}
& & & & 
\multirow{2}{*}{ \small MCL1 } 
          & \small $\,\,\,\,\,\,\,\,\xi=1$ & \small 60\%$\,\,\,\,\,\,\,\,\,\,$ & \small $\,\,\,\,\,\,\,\,\,\,\,\,\,\xi=1$ & \small 47\%$\,\,\,\,\,\,\,\,\,\,$ & \small $\,\,\,\,\,\,\,\,\,\,\,\,\,\xi=1$ & \small 73\%$\,\,\,\,\,\,\,\,\,\,$ 
\\
& & & & 
          & \small $\,\,\,\,\,\,\,\,\xi=2$ & \small 100\%$\,\,\,\,\,\,\,\,\,\,$ & \small $\,\,\,\,\,\,\,\,\,\,\,\,\,\xi=2$ & \small 100\%$\,\,\,\,\,\,\,\,\,\,$ & \small $\,\,\,\,\,\,\,\,\,\,\,\,\,\xi=2$ & \small 100\%$\,\,\,\,\,\,\,\,\,\,$ 
\\
\cline{2-11}
&
{ \small Stimuli } 
&
{ \small apoptosis } 
&
\small $4$
&
{ \small GZMB } 
          & \small $\,\,\,\,\,\,\,\,\xi=1$ & \small 100\%$\,\,\,\,\,\,\,\,\,\,$ & \small $\,\,\,\,\,\,\,\,\,\,\,\,\,\xi=1$ & \small 100\%$\,\,\,\,\,\,\,\,\,\,$ & \small $\,\,\,\,\,\,\,\,\,\,\,\,\,\xi=1$ & \small 100\%$\,\,\,\,\,\,\,\,\,\,$ 
\\
\hline
\end{tabular}
}
\end{table*}
\renewcommand{\arraystretch}{1}
\renewcommand{\tabcolsep}{6pt}
%

\vspace*{0.1in}
\noindent
{\bf Empirical evaluation of $(\star\!\star)$} 

\vspace*{0.1in}
\noindent
We empirically investigated the size $\xi$ of the neighborhood in claim ($\star\!\star$) for 
the same two biological networks and the same combinations of source, target and central nodes as in claim ($\star$).
We considered the $\xi$-neighborhood of $u_{\mathrm{central}}$ for $\xi=1,2,\dots$, and for each such neighborhood 
we determined what percentage of shortest or approximately short path (with one or two extra edges compared to shortest paths) 
between $u_{\mathrm{source}}$ and $u_{\mathrm{target}}$
involved a node in this neighborhood
(not counting $u_{\mathrm{source}}$ and $u_{\mathrm{target}}$).
The results, tabulated in Table~\ref{KO}, show that
removing the nodes in a
$\xi\leq\gadd_{\mathrm{worst}}(G)$ 
neighborhood around the central nodes disrupts all the relevant paths of the selected networks.
As $\gadd_{\mathrm{worst}}(G)$ is a small constant for
all of our biological networks, this implies that the central node and its neighbors within a small distance are the essential nodes in the 
signal propagation between $u_{\mathrm{source}}$ and $u_{\mathrm{target}}$.


\subsection{Effect of hyperbolicity on structural holes in social networks}
\label{socialnet-sec}

For a node $u\in V$, let $\Nbr(u)=\left\{ \, v \,|\,\{u,v\}\in E \, \right\}$ be the set of neighbors of (\IE, nodes adjacent to) $u$.
To quantify the useful information in a social network, Ron Burt in~\cite{Burt95} defined a measure of the {\em structural holes} of a network. 
For an undirected unweighted connected graph $G=(V,E)$ and a node $u\in V$ with degree larger than $1$, this measure
$\mathfrak{M}_u$ of the structural hole at $u$ is defined as~\cite{Burt95,Borg97}:

\begin{multline*}
\mathfrak{M}_u
\stackrel{\mathrm{def}}{=\joinrel=}
\sum_{v\in V} 
\,
\left(
\,
\dfrac
{
a_{u,v}+a_{v,u}
}
{
\displaystyle
\max_{x\neq u} 
\Big\{ \, a_{u,x}+a_{x,u} \, \Big\}
}
\text{
\raisebox{-1.35ex} { \scalebox{1}[4] {$[$}  }
}
\!\!
1 \,- 
\right.
\\
\left.
\left.
\!\!
\text
{
\footnotesize
$
\displaystyle
\sum_{ \substack{y\in V \\ y\neq u,v} } 
$
}
\!\!
\left(
\frac
{
a_{u,y}+a_{y,u}
}
{
\!\!
\text
{
\footnotesize
$
\displaystyle
\sum_{x\neq u} 
$
}
\!\!
\left( a_{u,x}+a_{x,u} \right)
}
\right)
\,
\left(
\frac
{
a_{v,y}+a_{y,v}
}
{
\!\!
\text
{
\footnotesize
$
\displaystyle
\max_{z\neq y} 
$
}
\!\!
\left\{ a_{v,z}+a_{z,v} \right\}
}
\right)
\,
\,
\right]
\,
\,
\right)
\end{multline*}
where 
$
a_{p,q}=\left\{
\begin{array}{ll}
1, & \mbox{if $\{p,q\}\in E$}
\\
0, & \mbox{otherwise}
\end{array}
\right.
$
are the entries in the standard adjacency matrix of $G$.
By observing that $a_{p,q}=a_{q,p}$ and 
$\displaystyle\max_{x\neq u} \big\{ a_{u,x}+a_{x,u} \big\}= \max_{z\neq y} \big\{ a_{v,z}+a_{z,v} \big\}=2$, 
the above equation for 
$\mathfrak{M}_u$
can be simplified to
\begin{gather}
\mathfrak{M}_u
=
\big|\,\Nbr(u) \,\big|
\,-
\dfrac
{
\displaystyle
\sum_{ 
v,y\,\in\, \Nbr(u)
} 
\!\!
\!\!
a_{v,y}
}
{
\big|\,\Nbr(u) \,\big|
}
\label{sh2}
\end{gather}
Thus high-degree nodes whose neighbors are not connected to each other have high $\mathfrak{M}_u$ values.
For an intuitive interpretation and generalization of \eqref{sh2}, the following definition of weak and strong dominance will prove 
useful ({\em cf}. dominating set problem for graphs~\cite{GJ79} and point domination problems in geometry~\cite{BDG1}).
A pair of distinct nodes $v,y$ is weakly $(\rho,\lambda)$-dominated (respectively, {\bf strongly} $(\rho,\lambda)$-dominated) 
by a node $u$ provided (see {\rm\FI{domin-fig}}):
\begin{enumerate}
\item[(a)]
$\rho< d_{u,v},d_{u,y}\leq\rho+\lambda$, and 

\item[(b)]
for at least one shortest path $\cP$ (respectively, {\bf for every shortest path} $\cP$)
between $v$ and $y$, $\cP$ contains a node $z$ such that $d_{u,z}\leq\rho$.
\end{enumerate}
%

%
\begin{figure}
\epsfig{file=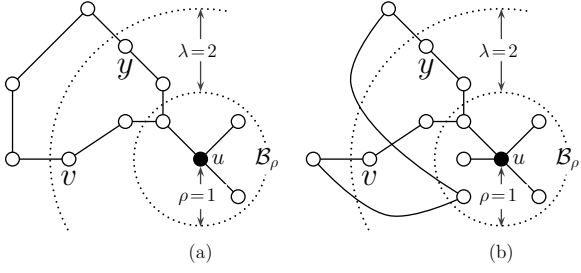}
\caption{\label{domin-fig}Illustration of weak and strong domination. 
(a) $v,y$ is weakly $(\rho,\lambda)$-dominated by $u$ since only one shortest path between $v$ and $y$ intersects $\mathcal{B}_{\rho}(u)$.
(b) $v,y$ is strongly $(\rho,\lambda)$-dominated by $u$ since all the shortest path between $v$ and $y$ intersect $\mathcal{B}_{\rho}(u)$.
}
\end{figure}
%

\noindent
Let 
$\mathbf{\{v,y\} \prec_{\,\mathrm{weak}}^{\,\rho,\lambda} u} \mbox{ (respectively, $\mathbf{\{v,y\} \prec_{\,\mathrm{strong}}^{\,\rho,\lambda} u}$) }$
\[
\hspace*{0.7in}
=
\left\{
\begin{array}{ll}
1, & 
\!\!
\begin{array}{l}
\mbox{\small if $v,y$ is weakly (respectively, {\bf strongly})}
\\
\,\,\,\,
\mbox{\small $(\rho,\lambda)$-dominated by $u$}
\end{array}
\\
[0.1in]
0, & \mbox{\small otherwise}
\end{array}
\right.
\]
Since $\mathcal{B}_1(u)=\bigcup_{0\,<\,j\,\leq\, 1} \mathcal{B}_j(u)=\Nbr(u)$, it follows that 
\[
\mathfrak{M}_u
=
\left| 
\,
\cup_{
0\,<\,j\,\leq\, 1} 
\mathcal{B}_j(u)
\,\right|
\,\,-\,\,
\frac
{
\sum_{ 
v,y\,\in\, 
\bigcup_{\,0\,<\,j\,\leq\, 1} \mathcal{B}_j(u)
} 
\left( 1 - \mathbf{\{v,y\} \prec_{\,\mathrm{weak}}^{\,0,1} u} \right)
}
{
\left| \, \cup_{0\,<\,j\,\leq\, 1} \mathcal{B}_j(u) \, \right|
}
\]
\vspace*{-10pt}
\[
=
\eX{
\left.
\hspace*{-0.15in}
\text{
\begin{tabular}{p{1.4in}}
\small 
number of pairs of nodes $v,y$ such that $v,y$ is {\bf weakly} $(0,1)$-dominated by $u$
\end{tabular}
}
\hspace*{-0.1in}
\right| 
\hspace*{-0.1in}
\text{
\begin{tabular}{p{1.4in}}
\small 
$v$ is selected uniformly randomly from $\bigcup_{0\,<\,j\,\leq\, 1} \mathcal{B}_j(u)$
\end{tabular}
}
\hspace*{-0.1in}
}
\]
\vspace*{-10pt}
\[
\geq
\eX{
\left.
\hspace*{-0.15in}
\text{
\begin{tabular}{p{1.4in}}
\small 
number of pairs of nodes $v,y$ such that $v,y$ is {\bf strongly} $(0,1)$-dominated by $u$
\end{tabular}
}
\hspace*{-0.1in}
\right| 
\hspace*{-0.1in}
\text{
\begin{tabular}{p{1.4in}}
\small 
$v$ is selected uniformly randomly from $\bigcup_{0\,<\,j\,\leq\, 1} \mathcal{B}_j(u)$
\end{tabular}
}
\hspace*{-0.1in}
}
\]
and a generalization of $\mathfrak{M}_u$ is given by (replacing $0,1$ by $\rho,\lambda$):
\[
\mathfrak{M}_{u,\rho,\lambda}
=
\left| 
\,
\cup_{
\rho\,<\,j\,\leq\, \lambda} 
\mathcal{B}_j(u)
\,\right|
\,\,-\,\,
\frac
{
\sum_{ 
v,y\,\in\, 
\bigcup_{\,\rho\,<\,j\,\leq\, \lambda} \mathcal{B}_j(u)
} 
\left( 1 - \mathbf{\{v,y\} \prec_{\,\mathrm{weak}}^{\,\rho,\lambda} u} \right)
}
{
\left| \, \bigcup_{\rho\,<\,j\,\leq\, \lambda} \mathcal{B}_j(u) \, \right|
}
\]
\vspace*{-10pt}
\[
=
\eX{
\hspace*{-0.15in}
\text{
\begin{tabular}{p{1.4in}}
\small 
number of pairs of nodes $v,y$ such that $v,y$ is {\bf weakly} $(\rho,\lambda)$-dominated by $u$
\end{tabular}
}
\hspace*{-0.1in}
\Bigg| 
\hspace*{-0.1in}
\text{
\begin{tabular}{p{1.4in}}
\small 
$v$ is selected uniformly randomly from $\cup_{\rho\,<\,j\,\leq\, \lambda} \mathcal{B}_j(u)$
\end{tabular}
}
\hspace*{-0.1in}
}
\]
\vspace*{-10pt}
\[
\geq
\eX{
\hspace*{-0.15in}
\text{
\begin{tabular}{p{1.4in}}
\small 
number of pairs of nodes $v,y$ such that $v,y$ is {\bf strongly} $(\rho,\lambda)$-dominated by $u$
\end{tabular}
}
\hspace*{-0.1in}
\Bigg| 
\hspace*{-0.1in}
\text{
\begin{tabular}{p{1.4in}}
\small 
$v$ is selected uniformly randomly from $\cup_{\rho\,<\,j\,\leq\, \lambda} \mathcal{B}_j(u)$
\end{tabular}
}
\hspace*{-0.1in}
}
\]
When the graph is hyperbolic (\IE, $\gadd_{\mathrm{worst}}(G)$ is a constant),  
for moderately large $\lambda$, weak and strong dominance are essentially identical and therefore weak domination has a much stronger implication.
Recall that $n$ denotes the number of nodes in the graph $G$.

%
\begin{figure}
\epsfig{file=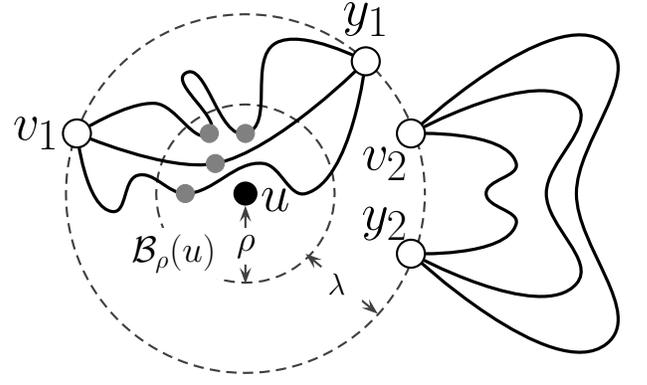}
\caption{\label{f12-prev}Visual illustration: either all the shortest paths are completely inside or all
the shortest paths are completely outside of $\mathcal{B}_{\rho+\lambda}(u)$.}
\end{figure}

Our finding can be succinctly summarized as (see \FI{f12-prev} for a visual illustration):

\begin{quote}
($\star\!\!\!\star\!\!\!\star$)
If $\lambda\geq \left( 6\,\gadd_{\mathrm{worst}}(G)+2 \right) \log_2 n$ then, assuming 
$v$ is selected uniformly randomly from $\cup_{\rho\,<\,j\,\leq\, \lambda} \mathcal{B}_j(u)$ for any node $u$, 
the expected number of pair of nodes $v,y$ that are {\bf weakly} $(\rho,\lambda)$-dominated by $u$
is precisely the same as 
the expected number of pair of nodes that are {\bf strongly} $(\rho,\lambda)$-dominated by $u$.
\end{quote}

\begin{figure}
\epsfig{file=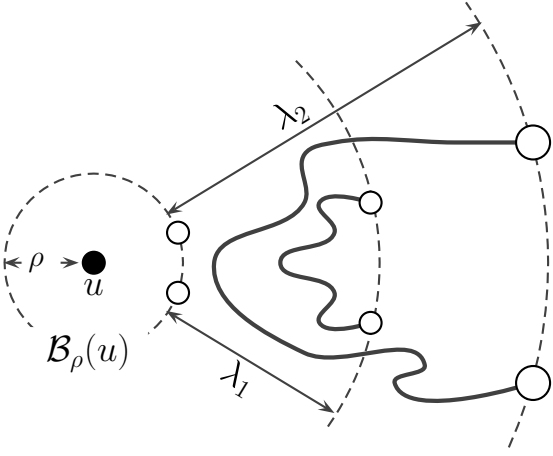}
\caption{\label{f12}For hyperbolic graphs, the further we move from the central (black) node, the more a shortest path bends inward towards the central node.}
\end{figure}


\noindent
A mathematical justification for the claim 
($\star\!\!\star\!\!\star$)
is provided by Lemma~\ref{lem4} in Section~\ref{lem4-sec} of the appendix.

\vspace*{0.1in}
\noindent
{\bf An implication of ($\star\!\!\star\!\!\star$)}

\begin{quote}
If $\lambda\geq \left( 6\,\gadd_{\mathrm{worst}}(G)+2 \right) \log_2 n$ and 
$\mathfrak{M}_{u,\rho,\lambda}\approx \big|\,\mathcal{B}_{\rho+\lambda}(u) \,\big|$, then {\em almost all} pairs of nodes 
are strongly $(\rho,\lambda)$-dominated by $u$, \IE, 
for almost all pairs of nodes $v,y\in\,\mathcal{B}_{\rho+\lambda}(u)$, every shortest path between $v$ and $y$
contains a node in $\mathcal{B}_{\rho}(u)$. 
\end{quote}

\noindent
A visual illustration of this implication is in \FI{f12} showing that 
{\em as $\lambda$ increases the shortest paths tend to 
bend more and more towards the central node $u$ for a hyperbolic network}.

%
\renewcommand{\arraystretch}{1}
\renewcommand{\tabcolsep}{4pt}
\begin{table}
\caption[]{\label{social-table1}Weak domination leads to strong domination for social networks.
$u$ is the index of the central node and 
\\
\hspace*{0.4in}
$
\nu=
\dfrac{n_2}{n_1}
=
\dfrac
{
\left|\,
\left\{
(v,y) 
\in
\mathcal{B}_{\rho+\lambda}(u)
\,\big|\,
\mathbf{\{v,y\} \prec_{\,\mathrm{strong}}^{\,\rho,\lambda} u}=1
\right\}
\,\right|
}
{
\left|\,
\left\{
(v,y) 
\in
\mathcal{B}_{\rho+\lambda}(u)
\,\big|\,
\mathbf{\{v,y\} \prec_{\,\mathrm{weak}}^{\,\rho,\lambda} u}=1
\right\}
\,\right|
}
$
}
\scalebox{0.95}[0.95]
{
\begin{tabular} { l r c c c r  }
\\
\hline
\multicolumn{1}{c}{\small Network name} &  \multicolumn{1}{c}{\small $u$} & \small $\rho$ & \small $\lambda$ & \small $\left| \, \mathcal{B}_{\rho+\lambda}(u)\, \right|$ & \multicolumn{1}{c}{\small $\nu$}
\\
\hline
\multirow{2}{*}{\small Network 1: Dolphin social network} & \small 14 & \small 4 & \small 1 & \small 5 & {\bf \small 80\%}
\\
& \small 37 & \small 4 & \small 1 & \small 3 & {\bf \small 100\%}
\\
\hline
\multirow{2}{*}{\small Network 4: Books about US politics} &  \small 8 & \small 4 & \small 1 & \small 4 & {\bf \small 83\%}
\\
&  \small 3 & \small 3 & \small 1 & \small 5 & {\bf \small 90\%}
\\
\hline
\multirow{2}{*}{\small Network 7: Visiting ties in San Juan} &  \small 34 & \small 4 & \small 1 & \small 4 & {\bf \small 50\%}
\\
&  \small 9 & \small 3 & \small 1 & \small 5 & {\bf \small 90\%}
\\
\hline
\end{tabular}
}
\end{table}
\renewcommand{\arraystretch}{1}
\renewcommand{\tabcolsep}{6pt}
%

\vspace*{0.1in}
\noindent
{\bf Empirical verification of 
($\star\!\!\star\!\!\star$)
}

\vspace*{0.1in}
\noindent
We empirically investigated the claim in 
($\star\!\!\star\!\!\star$)
for the following three social networks from Table~\ref{L4}:

\begin{quote}
\begin{description}
\item[Network 1]  
Dolphin social network,

\item[Network 4] 
Books about US politics, and  

\item[Network 7] 
Visiting ties in San Juan.
\end{description}
\end{quote}

\noindent
For each network we selected a (central) node $u$ such that there are sufficiently many nodes in the boundary of the $\xi$-neighborhood $\mathcal{B}_{\xi}\,(u)$ of $u$ 
for an appropriate $\xi=\rho+\lambda$. We then set $\lambda$ to a very small value of $1$, and calculated the following quantities.
\begin{itemize}
\item
We computed the number $n_1$ of all pairs of nodes from $\mathcal{B}_{\xi}\,(u)$ that are weakly $(\rho,\lambda)$-dominated by $u$.

\item
We computed the number $n_2$ of all pairs of nodes from $\mathcal{B}_{\xi}\,(u)$ that are strongly $(\rho,\lambda)$-dominated by $u$.
\end{itemize}
Table~\ref{social-table1} tabulates the ratio $\nu=\nicefrac{n_2}{n_1}$, and shows that 
a large percentage of the pair of nodes that were weakly dominated were
also strongly dominated by $u$.

\section{Conclusion}

In this paper we demonstrated a number of interesting properties of the shortest and approximately shortest paths in hyperbolic networks. We established
the relevance of these results in the context of biological and social networks by empirically finding that a variety of such networks have close-to-tree-like 
topologies. Our results have important implications to a general class of directed networks which we refer to as regulatory networks. For example,
our results imply that cross-talk edges or paths are frequent in these networks. Based on our theoretical results we proposed methodologies to determine
relevant paths between a source and a target node in a signal transduction network, and to identify the most important nodes that mediate these paths.
Our investigation shows that the hyperbolicity measure captures non-trivial topological properties that is not fully reflected in other network measures, and therefore 
the hyperbolicity measure should be more widely used.

\appendix

\section{Theorem~\ref{kk1}}
\label{kk1-sec}

\begin{theorem}\label{kk1}
Suppose that $G$ has a cycle of $k\geq 4$ nodes which has no path-chord. Then, $\gadd_{\mathrm{worst}}(G) \geq \left\lceil \nicefrac{k}{4}\right\rceil $.
\end{theorem}


\begin{figure}[b]
\epsfig{file=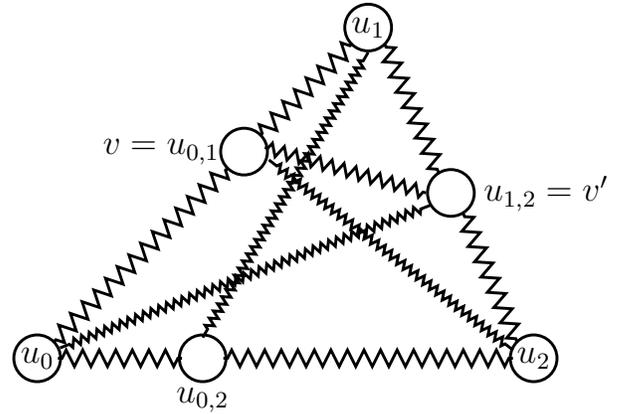}
\caption{\label{f5}Case 1 of Theorem~\ref{lem1}: $v=u_{0,1}$, $v'=u_{1,2}$.}
\end{figure}

\begin{proof}
In our proofs we will use the consequences of the $4$-node condition when the $4$ nodes are chosen in a specific manner as stated below
in Lemma~\ref{pre1}.

\begin{lemma}\label{pre1}
Let $u_0,u_1,u_2,u_3$ be four nodes such that $u_3$ is on a shortest path between $u_1$ and $u_2$. Suppose also that all the inter-node 
distances are {\em strictly} positive except for $d_{u_1,u_3}$ and 
$
d_{u_1,u_3}
=
\left
\lceil
\frac{ d_{u_1,u_2} +  d_{u_0,u_1} -  d_{u_0,u_2}}{2}
\right
\rceil
$.
Then, 
\begin{multline*}
\left\lceil \frac{ d_{u_0,u_1}+d_{u_0,u_2}+d_{u_1,u_2} }{2} \right\rceil
\leq
d_{u_0,u_3}+d_{u_1,u_2}
\\
\leq
\left\lceil \frac{ d_{u_0,u_1}+d_{u_0,u_2}+d_{u_1,u_2} }{2} \right\rceil
+
2\,\gadd_{u_0,u_1,u_2,u_3}
\end{multline*}
\end{lemma}

\begin{proof}
Note that due to triangle inequality
$
0 \leq 
\left
\lceil
\frac{ d_{u_1,u_2} +  d_{u_0,u_1} -  d_{u_0,u_2}}{2}
\right
\rceil
\leq d_{u_1,u_2} 
$
and thus node $u_3$ always exists.

First, consider the case when $0<d_{u_1,u_3}<d_{u_1,u_2}$.
Consider the three quantities involved in the $4$-node condition for the nodes $u_0,u_1,u_2,u_3$, namely the quantities
$
d_{u_0,u_3}+d_{u_1,u_2}
$,
$
d_{u_0,u_2}+d_{u_1,u_3}
$ 
and
$
d_{u_0,u_1}+d_{u_2,u_3}
$.
Note that 
\begin{multline*}
2 \left(d_{u_0,u_3}+d_{u_1,u_2} \right)
=
\left( d_{u_0,u_3} + d_{u_1,u_3} \right) 
+
\left( d_{u_0,u_3} + d_{u_2,u_3} \right) 
+
d_{u_1,u_2}
\\
\geq
d_{u_0,u_1}+d_{u_0,u_2}+d_{u_1,u_2}
\\
\Rightarrow 
\,
d_{u_0,u_3}+d_{u_1,u_2}
\geq
\left\lceil \frac{
d_{u_0,u_1}+d_{u_0,u_2}+d_{u_1,u_2}
}{2} \right\rceil
\end{multline*}
%
%
\begin{multline*}
d_{u_0,u_2}+d_{u_1,u_3}
=
d_{u_0,u_2}
+
\left\lfloor \frac{ d_{u_1,u_2} +  d_{u_0,u_1} -  d_{u_0,u_2}}{2} \right\rfloor
\\
=
\left\lfloor \frac{
d_{u_0,u_1}+d_{u_0,u_2}+d_{u_1,u_2}
}{2} \right\rfloor
\end{multline*}
%
%
\begin{multline*}
d_{u_0,u_1}+d_{u_2,u_3}
=
d_{u_0,u_1}+
\left\lceil
\frac{ d_{u_1,u_2} +  d_{u_0,u_2} -  d_{u_0,u_1} }{2}
\right\rceil
\\
=
\left\lceil
\frac{
d_{u_0,u_1}+d_{u_0,u_2}+d_{u_1,u_2}
}{2} 
\right\rceil
\end{multline*}
Thus, 
$
d_{u_0,u_3}+d_{u_1,u_2} \geq \max \big\{ \,d_{u_0,u_2}+d_{u_1,u_3}, \, d_{u_0,u_1}+d_{u_2,u_3} \,\big\} 
$
and using the definition of $\gadd_{u_0,u_1,u_2,u_3}$ we have 
\begin{multline*}
\left\lceil \frac{ d_{u_0,u_1}+d_{u_0,u_2}+d_{u_1,u_2} }{2} \right\rceil
\leq
d_{u_0,u_3}+d_{u_1,u_2}
\\
\leq
\left\lceil \frac{ d_{u_0,u_1}+d_{u_0,u_2}+d_{u_1,u_2} }{2} \right\rceil
+
2\,\gadd_{u_0,u_1,u_2,u_3}
\end{multline*}
Next, consider the case when $d_{u_1,u_3}=0$.
This implies 
\begin{multline*}
d_{u_0,u_1} + d_{u_1,u_3}
=
d_{u_0,u_1} + d_{u_1,u_2}
=
d_{u_0,u_2}
\\
=
\dfrac{ d_{u_0,u_1}+d_{u_0,u_2}+d_{u_1,u_2} }{2}
\leq
\left\lceil \frac{ d_{u_0,u_1}+d_{u_0,u_2}+d_{u_1,u_2} }{2} \right\rceil
\end{multline*}
Finally, consider the case when $d_{u_1,u_3}=d_{u_1,u_2}$.
This implies 
\begin{multline*}
d_{u_1,u_2}
-
\frac{ d_{u_1,u_2} +  d_{u_0,u_1} -  d_{u_0,u_2}}{2}
< 1
\\
\equiv \,
d_{u_0,u_2} + d_{u_1,u_2} = d_{u_0,u_1} + 2 - 2\,\eps\,\,\, \mbox{ for some $0<\eps\leq 1$} 
\end{multline*}
Thus, it easily follows that
\begin{multline*}
d_{u_0,u_3}+d_{u_1,u_2}
=
d_{u_0,u_2}+d_{u_1,u_2}
=
\frac
{
d_{u_0,u_2}+d_{u_1,u_2}
+
d_{u_0,u_1} + 2 - 2\,\eps
}
{2}
\\
=
\frac
{
d_{u_0,u_2}+d_{u_1,u_2} + d_{u_0,u_1} 
}
{2}
+ 1 - \eps
\\
\Rightarrow \, 
d_{u_0,u_3}+d_{u_1,u_2}
\leq
\left\lceil \frac{ d_{u_0,u_1}+d_{u_0,u_2}+d_{u_1,u_2} }{2} \right\rceil
\end{multline*}
\end{proof}

We can now prove Theorem~\ref{kk1} as follows.
Let $\cC=\big( u_0, u_1, \dots, u_{k-1},u_0 \big)$ be the cycle of $k=4r+r'$ nodes for some integers $r$ and $0\leq r'<4$. 
Consider the four nodes $u_0,u_{r+\left\lceil\nicefrac{r'}{2}\right\rceil},u_{2r+\left\lfloor \, \left( \, r'+ \left\lceil\nicefrac{r'}{2}\right\rceil  \, \right) \, / \, 2\,\right\rfloor}$ and
$u_{3r+r'}$. Since $\cC$ has no path-chord, we have 
$d_{u_0,u_{r+\left\lceil\nicefrac{r'}{2}\right\rceil}}=r+\left\lceil\nicefrac{r'}{2}\right\rceil$,
$d_{u_0,u_{2r+  \left\lfloor \, \left( \, r'+ \left\lceil\nicefrac{r'}{2}\right\rceil  \, \right) \, / \, 2\,\right\rfloor} }=2r+ \left\lfloor \, \frac{ r'+ \left\lceil\nicefrac{r'}{2}\right\rceil  } { 2 } \,\right\rfloor d_{u_{r+\left\lceil r' / 2 \right\rceil},u_{3r+r'}}= 2r + r' - \left\lceil\nicefrac{r'}{2}\right\rceil \leq 2r+ \left\lceil\nicefrac{r'}{2}\right\rceil$,
$d_{u_0,u_{3r+r'}} = r$
and $u_{2r+  \left\lfloor \, \left( \, r'+ \left\lceil\nicefrac{r'}{2}\right\rceil  \, \right) \, / \, 2\,\right\rfloor  }$ is on a shortest path between $u_r$ and $u_{3r+r'}$. 
Thus, applying the bound of Lemma~\ref{pre1}, we get
\begin{multline*}
\gadd_{\mathrm{worst}}(G) 
\geq 
\gadd_{
u_0,\,u_{r+\left\lceil\nicefrac{r'}{2}\right\rceil},\,u_{2r+\left\lfloor \, \left( \, r'+ \left\lceil\nicefrac{r'}{2}\right\rceil  \, \right) \, / \, 2\,\right\rfloor},\, 
u_{3r+r'}
}
\\
\geq 
\frac
{
d_{u_0,u_{2r+  \left\lfloor \, \frac{r'+ \left\lceil\nicefrac{r'}{2}\right\rceil } {2}\right\rfloor          }}+d_{u_{r+\left\lceil\frac{r'}{2}\right\rceil},u_{3r+r}}
-
\left\lceil \frac{ d_{u_0,u_{r+\left\lceil\frac{r'}{2}\right\rceil}}+d_{u_{r+\left\lceil\frac{r'}{2}\right\rceil},u_{3r+r'}}+d_{u_{3r+r},u_0} }{2} \right\rceil
}
{2}
\\
=
\frac {4r + 
\left\lfloor \, \frac{ r'+ \left\lceil\nicefrac{r'}{2}\right\rceil  } { 2 } \,\right\rfloor
- r' + \left\lceil\nicefrac{r'}{2}\right\rceil  - \left\lceil \dfrac{ 4r+r'}{2} \right\rceil } {2}
=
r+ \frac {
\left\lfloor \, \frac{ r'+ \left\lceil\nicefrac{r'}{2}\right\rceil  } { 2 } \,\right\rfloor
- r'
} {2}
\\
\geq
r - \nicefrac{1}{4}
\,\,\Rightarrow \,\,
\gadd_{\mathrm{worst}}(G) 
\geq 
r = \left\lceil \nicefrac{k}{4} \right\rceil
\end{multline*}
\end{proof}



\section{Theorem~\ref{lem1} and Corollary~\ref{lem1-cor}}
\label{lem1-sec}

\begin{figure}
\epsfig{file=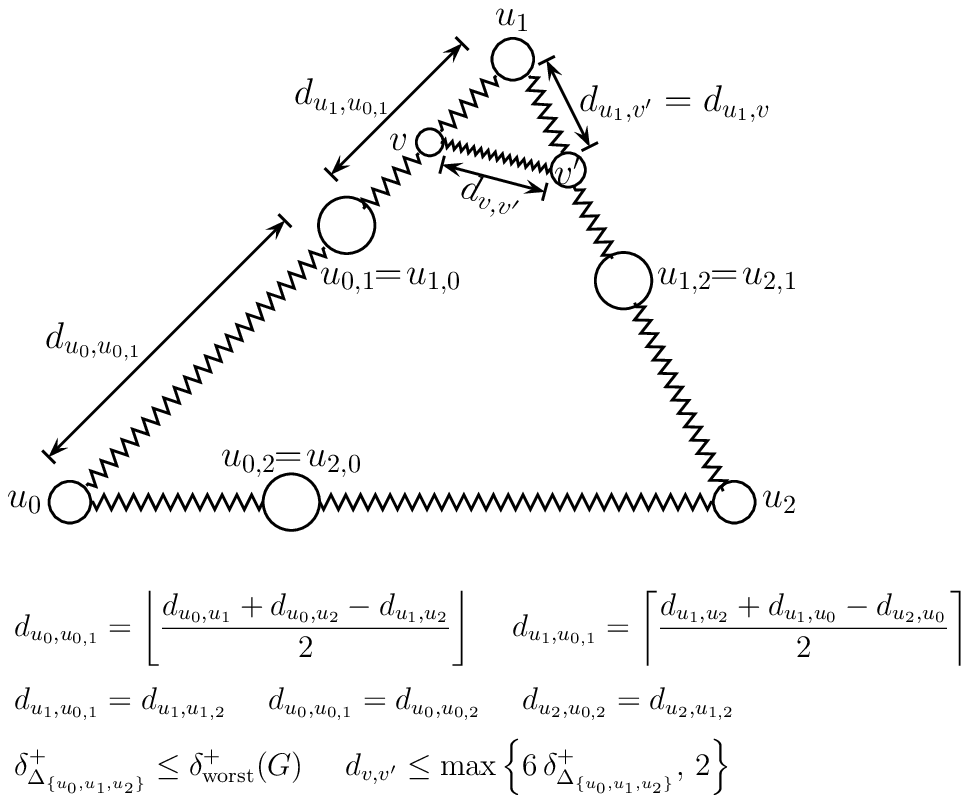,scale=0.9}
\caption{\label{f4}A pictorial illustration of the claim in Theorem~\ref{lem1}.}
\end{figure}
%
%
%
%
%

The Gromov product nodes $u_{0,1},u_{0,2},u_{1,2}$ of a shortest-path triangle $\Delta_{\left\{u_0,u_1,u_2\right\}}$ 
are three nodes satisfying the following\footnote{To simplify exposition, we assume that $d_{u_0,u_1}+d_{u_1,u_2}+d_{u_0,u_2}$ is 
an even number. Otherwise, the definition will require minor changes.}:
\begin{itemize}
\item
$u_{0,1}$, $u_{0,2}$ and $u_{1,2}$ 
are located on the paths 
$\cP_{\Delta}\left(u_0,u_1\right)$, 
$\cP_{\Delta}\left(u_0,u_2\right)$ and 
$\cP_{\Delta}\left(u_1,u_2\right)$, respectively, and 

\item
the distances of these three nodes from $u_0,u_1$ and $u_2$ satisfy the following constraints:
\begin{gather*}
d_{u_0,u_{0,1}}+d_{u_1,u_{0,1}}=d_{u_0,u_1}, 
\,\,\,\,
d_{u_0,u_{0,2}}+d_{u_2,u_{0,2}}=d_{u_0,u_2}
\\
d_{u_1,u_{1,2}}+d_{u_2,u_{1,2}}=d_{u_1,u_2},
\,\,\,\,
d_{u_1,u_{0,1}}= d_{u_1,u_{1,2}}
\\
d_{u_0,u_{0,1}}=d_{u_0,u_{0,2}}=\left\lfloor\dfrac{d_{u_0,u_1}+d_{u_0,u_2}-d_{u_1,u_2}}{2}\right\rfloor
\end{gather*}
\end{itemize}
It is not difficult to see that a set of such three nodes always exists.
For convenience, the nodes $u_{1,0}$, $u_{2,0}$ and $u_{2,1}$ are assumed to be the same as
the nodes $u_{0,1}$, $u_{0,2}$ and $u_{1,2}$, respectively. 

\begin{theorem}[see \FI{f4} for a visual illustration]
\label{lem1}
For a shortest-path triangle $\Delta_{\left\{u_0,u_1,u_2\right\}}$ and for $0\leq i\leq 2$,
let $v$ and $v'$ be two nodes on the paths 
$u_i \!\!\!\stackrel{ \cP_{\Delta}\left(u_i,u_{i+2 \pmod{3}}\right) } {\text{\scalebox{5}[1.2]{$\leftrightsquigarrow$} }  }\!\!\! u_{i,\,i+2 \pmod{3}}$ 
and 
$u_i \!\!\!\stackrel{ \cP_{\Delta}\left(u_i,u_{i+1 \pmod{3}}\right) } {\text{\scalebox{5}[1.2]{$\leftrightsquigarrow$} }  }\!\!\! u_{i,\,i+1 \pmod{3}}$, 
respectively, such that $d_{u_i,v}=d_{u_i,v'}$.
Then, 
\[
d_{v,v'}\leq 6\,\gadd_{ \Delta_{ \{u_0,u_1,u_2\} }} + 2
\]
where
$\gadd_{ \Delta_{ \{u_0,u_1,u_2 \} } } \leq \gadd_{\mathrm{worst}}(G)$ 
is the largest worst-case hyperbolicity among all combinations of four nodes in the three shortest paths defining the triangle. 
\end{theorem}

\begin{corollary}[Hausdorff distance between shortest paths]\label{lem1-cor}
Suppose that $\cP_1$ and $\cP_2$ are two shortest paths between two nodes $u_0$ and $u_1$. Then, 
the Hausdorff distance $d_H\left(\cP_1,\cP_2\right)$ between these two paths can be bounded as:
\begin{multline*}
\!\!\!\!\!
\!\!
d_H\left(\cP_1,\cP_2\right)
\eqdef
\max 
\left\{ \,
\max_{v_1\,\in \,\cP_1} \min_{v_2\,\in \,\cP_2} \Big\{\, d_{v_1,v_2} \,\Big\}
, \, 
\max_{v_2\,\in \,\cP_2} \min_{v_1\,\in \,\cP_1} \Big\{\, d_{v_1,v_2} \,\Big\}
\, \right\}
\\
\leq
6\,\gadd_{ \Delta_{ \{u_0,u_1,u_2\} }} + 2
\end{multline*}
{where $u_2$ is any node on the path $\cP_2$.}
\end{corollary}

\noindent
{\bf Proof of Theorem~\ref{lem1}.}
To simplify exposition, we assume that $d_{u_0,u_1}+d_{u_1,u_2}+d_{u_0,u_2}$ is 
even and prove a slightly improve bound of 
$d_{v,v'}\leq 6\,\gadd_{ \Delta_{ \{u_0,u_1,u_2\} }} + 1$.
It is easy to modify the proof to show that
$d_{v,v'}\leq 6\,\gadd_{ \Delta_{ \{u_0,u_1,u_2\} }} + 2$
if $d_{u_0,u_1}+d_{u_1,u_2}+d_{u_0,u_2}$ is odd.

We will prove the result for $i=1$ only; similar arguments will hold for $i=0$ and $i=2$.
If $d_{u_1,u_{0,1}}=0$ then $v=v'=u_1$ and the claim holds trivially, Thus, we assume that $d_{u_1,u_{0,1}}>0$.

\vspace*{0.1in}
\noindent
{\bf Case 1: ${v=u_{0,1}}$ and ${v'=u_{1,2}}$}. 
In this case we need to prove that $d_{u_{0,1},u_{1,2}}\leq 6\,\gadd_{\Delta_{\{u_0,u_1,u_2\}}}+1$ (see \FI{f5}).
Assume that $d_{u_{0,1},u_{1,2}}>0$ since otherwise the claim is trivially true.
Using Lemma~\ref{pre1} for the four nodes $u_0,u_1,u_2,u_{1,2}$, we get 
\begin{gather}
d_{u_0,u_{1,2}}+d_{u_1,u_2} 
\leq 
\left\lceil \frac{
d_{u_0,u_1}+d_{u_1,u_2}+d_{u_0,u_2}
}{2} \right\rceil
+ 2\,\gadd_{  u_0,u_1,u_2,u_{1,2} }
\label{j6}
\end{gather}
Now, we note that
\begin{multline}
d_{u_1,u_2}+d_{u_0,u_{0,2}}
=
d_{u_1,u_2}+
\left\lfloor \frac{ d_{u_0,u_1} +  d_{u_0,u_2} -  d_{u_1,u_2} }{2} \right\rfloor
\\
=
\left\lfloor \frac{ d_{u_0,u_1} +  d_{u_0,u_2} +  d_{u_1,u_2} }{2} \right\rfloor
\label{j7}
\end{multline}
which in turn implies 
\begin{multline}
\big| \, d_{u_0,u_{1,2}} - d_{u_0,u_{0,2}} \, \big| 
=
\left| \, \left( d_{u_0,u_{1,2}}+d_{u_1,u_2} \right) - \left( d_{u_1,u_2}+d_{u_0,u_{0,2}} \right) \, \right| 
\\
\leq 
\Bigg| \,\, 
\underbrace{  
\left\lceil \frac{
d_{u_0,u_1}+d_{u_1,u_2}+d_{u_0,u_2}
}{2} \right\rceil
+ 2\,\gadd_{  u_0,u_1,u_2,u_{1,2} }
}_{ \text{(by inequality \eqref{j6})} }
\\
\hspace*{1.3in}
\,-\,
\underbrace{ 
\left\lfloor \frac{ d_{u_0,u_1} +  d_{u_0,u_2} +  d_{u_1,u_2} }{2} \right\rfloor
}_{\text{(by equality \eqref{j7})} }
\,\, \Bigg|
\\
\leq 
2\,\gadd_{  u_0,u_1,u_2,u_{1,2} }
+1
\label{eq1}
\end{multline}
In a similar manner, we can prove the following analog of inequality~\eqref{eq1}:
\begin{gather}
\big| \, d_{u_2,u_{0,1}} - d_{u_2,u_{0,2}} \, \big| 
\leq 
2\,\gadd_{  u_0,u_1,u_2,u_{0,1} }
\label{eq2}
\end{gather}
Using inequalities~\eqref{eq1} and \eqref{eq2}, it follows that
\begin{multline}
\hspace*{-0.1in}
\left| \, \left( d_{u_0,u_{1,2}} + d_{u_2,u_{0,1}} \right) - d_{u_0,u_2} \, \right|
\\
=
\left| \, \left( d_{u_0,u_{1,2}} + d_{u_2,u_{0,1}} \right) - \left( d_{u_0,u_{0,2}} + d_{u_2,u_{0,2}} \right) \, \right|
\\
=
\left| \, \left( d_{u_0,u_{1,2}} - d_{u_0,u_{0,2}} \right) + \left( d_{u_2,u_{0,1}} - d_{u_2,u_{0,2}} \right) \, \right| 
\\
\leq
\big| \, d_{u_0,u_{1,2}} - d_{u_0,u_{0,2}} \, \big| \, + \, \big| \, d_{u_2,u_{0,1}} - d_{u_2,u_{0,2}} \, \big| 
\\
\leq
2\,\gadd_{  u_0,u_1,u_2,u_{1,2} }
+
2\,\gadd_{  u_0,u_1,u_2,u_{0,1} }
+1
\label{eq3}
\end{multline}
Now, consider the three quantities involved in the $4$-node condition for the nodes $u_0,u_2,u_{0,1},u_{1,2}$, namely the quantities:
$d_{u_0,u_2}+d_{u_{0,1},u_{1,2}}$, $d_{u_0,u_{1,2}}+d_{u_{0,1},u_2}$ and $d_{u_0,u_{0,1}}+d_{u_2,u_{1,2}}$.
Note that 
\begin{gather}
d_{u_0,u_{0,1}}\!\!  +d_{u_2,u_{1,2}}
\! =
d_{u_0,u_{0,2}}\!\!  +d_{u_2,u_{0,2}}
\! =
d_{u_0,u_2}
\! <
d_{u_0,u_2}\!\!  +d_{u_{0,1},u_{1,2}}
\label{j8}
\end{gather}
If 
$d_{u_0,u_{1,2}}+d_{u_{0,1},u_2}\leq d_{u_0,u_{0,1}}+d_{u_2,u_{1,2}}$ then by the definition of 
$\gadd_{  u_0,u_2,u_{0,1},u_{1,2} }$
we have 
\begin{multline*}
d_{u_{0,1},u_{1,2}}
=
\left( d_{u_0,u_2}+d_{u_{0,1},u_{1,2}} \right) - d_{u_0,u_2}
\\
=
\left( d_{u_0,u_2}+d_{u_{0,1},u_{1,2}} \right) - 
\left( d_{u_0,u_{0,1}}+d_{u_2,u_{1,2}} \right) 
\leq 
2\,\gadd_{  u_0,u_2,u_{0,1},u_{1,2} }
\end{multline*}
Otherwise, 
$d_{u_0,u_{1,2}}+d_{u_{0,1},u_2}>d_{u_0,u_{0,1}}+d_{u_2,u_{1,2}}$ and then again by the definition of 
$2\,\gadd_{  u_0,u_2,u_{0,1},u_{1,2} }$
we have 
\[
\left| \,
d_{u_0,u_{1,2}}+d_{u_{0,1},u_2}- d_{u_0,u_2} -d_{u_{0,1},u_{1,2}}
\, \right|
\leq 
2\,\gadd_{  u_0,u_2,u_{0,1},u_{1,2} }
\]
and now using inequality~\eqref{eq3} gives
\begin{multline*}
\!\!\!\!\!
\!\!\!\!\!
d_{u_{0,1},u_{1,2}}
=
\Big( d_{u_0,u_{1,2}} + d_{u_2,u_{0,1}} - d_{u_0,u_2} \Big)
-
\Big( d_{u_0,u_{1,2}}+d_{u_{0,1},u_2}- d_{u_0,u_2} -d_{u_{0,1},u_{1,2}} \Big)
\\
\leq 
\Big| d_{u_0,u_{1,2}} + d_{u_2,u_{0,1}} - d_{u_0,u_2} \Big|
+
\Big| d_{u_0,u_{1,2}}+d_{u_{0,1},u_2}- d_{u_0,u_2} -d_{u_{0,1},u_{1,2}} \Big|
\\
\leq
2\,\gadd_{  u_0,u_1,u_2,u_{1,2} }
+
2\,\gadd_{  u_0,u_1,u_2,u_{0,1} }
+
2\,\gadd_{  u_0,u_2,u_{0,1},u_{1,2} }
+1 
\leq
6\,\gadd_{\Delta_{\{u_0,u_1,u_2\}}}
+1
\end{multline*}
%

\begin{figure}
\epsfig{file=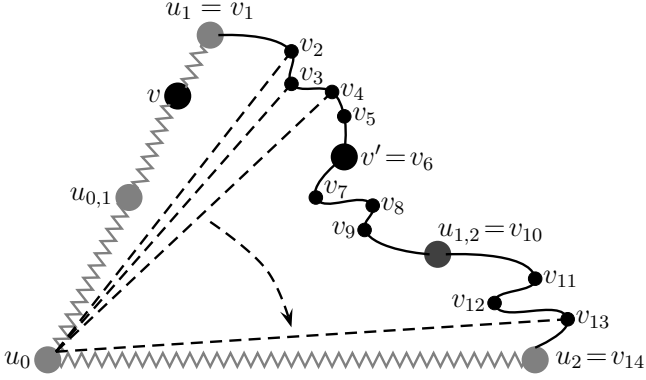}
\caption{\label{f6}Case 2 of Theorem~\ref{lem1}: $v\neq u_{0,1}$, $v'\neq u_{1,2}$.}
\end{figure}


\vspace*{0.1in}
\noindent
{\bf Case 2: ${v\neq u_{0,1}}$ and ${v'\neq u_{1,2}}$}. 
The claim trivially holds if $d_{v,v'}\leq 1$, thus we assume that $d_{v,v'}>1$.
Let $\big(v_1=u_1,v_2=u_3,v_3,\dots,v_{h}=v',\dots,v_s=u_{1,2},\dots,v_r=u_2\big)$ be the {\em ordered} sequence of nodes in the given 
shortest path from $u_1$ to $u_2$ (see \FI{f6}). 
Consider the sequence of shortest-path triangles $\Delta_{\left\{u_0,u_1,v_2\right\}},\Delta_{\left\{u_0,u_1,v_3\right\}},\dots,\Delta_{\left\{u_0,u_1,v_r\right\}}$, 
where each such triangle $\Delta_{\left\{u_0,u_1,v_j\right\}}$ is obtained by taking the shortest path 
$\cP_{\Delta}\left(u_0,u_1\right)$, 
the sub-path  
$\cP_{\Delta}\big(u_1,v_j\big)$
of the shortest path 
$\cP_{\Delta}\left(u_1,u_2\right)$, 
from $u_1$ to $v_j$, 
and a shortest path $u_0\!\stackrel{\mathfrak{s}}{\leftrightsquigarrow}\! v_j$ from $u_0$ to $v_j$.
Let $v_{1,j}$ be the Gromov product node on the side (shortest path) 
$\cP_{\Delta}\big(u_1,v_j\big)$
for the shortest-path triangle $\Delta_{\left\{u_0,u_1,v_j\right\}}$.

We claim that if $v_{1,j}=v_p$ and $v_{1,j+1}=v_q$ then $q$ is either $p$ or $p+1$. 
Indeed, if 
$
d_{u_1,v_p}
=
\left\lfloor
\frac { d_{u_0,u_1} + d_{u_1,u_j} - d_{u_0,v_j} } { 2 } 
\right\rfloor
$
and 
$
d_{u_1,v_q}
=
\left\lfloor
\frac { d_{u_0,u_1} + d_{u_1,u_{j+1}} - d_{u_0,v_{j+1}} } { 2 } 
\right\rfloor
$
then 
\begin{multline*}
\!\!\!\!\!
\!\!\!\!\!
d_{u_1,v_q}
-
d_{u_1,v_p}
=
\left\lfloor
\frac { d_{u_0,u_1} + d_{u_1,v_{j+1}} - d_{u_0,v_{j+1}} } { 2 } 
\right\rfloor
-
\left\lfloor
\frac { d_{u_0,u_1} + d_{u_1,v_j} - d_{u_0,v_j} } { 2 } 
\right\rfloor
\\
\leq
\left\lfloor
\frac { d_{u_0,u_1} + \left( 1+ d_{u_1,v_j} \right) - \left( d_{u_0,v_{j+1}} -1 \right) } { 2 } 
\right\rfloor
-
\left\lfloor
\frac { d_{u_0,u_1} + d_{u_1,v_j} - d_{u_0,v_j} } { 2 } 
\right\rfloor
\\
=
\left\lfloor
\frac { d_{u_0,u_1} + d_{u_1,v_j} - d_{u_0,v_j} } { 2 } +1
\right\rfloor
-
\left\lfloor
\frac { d_{u_0,u_1} + d_{u_1,v_j} - d_{u_0,v_j} } { 2 } 
\right\rfloor
\leq 1
\end{multline*}
and a similar proof of 
$d_{u_1,v_q} - d_{u_1,v_p} \leq 1$ can be obtained if 
$
d_{u_1,v_p}
=
\left\lceil
\frac { d_{u_0,u_1} + d_{u_1,u_j} - d_{u_0,v_j} } { 2 } 
\right\rceil
$
and 
$
d_{u_1,v_q}
=
\left\lceil
\frac { d_{u_0,u_1} + d_{u_1,u_{j+1}} - d_{u_0,v_{j+1}} } { 2 } 
\right\rceil
$.
Thus, the ordered sequence of nodes 
$v_{1,1},v_{1,2},\dots,v_{1,r}$ 
cover the ordered sequence of nodes 
$v_2,v_3,\dots,v_s$ in a {\em consecutive manner} without skipping over any node.
Since $v_{1,1}$ is either $v_1$ or $v_2$, 
and $v_{1,r}=v_s=u_{1,2}$, there must be an index $t$ such that $v_{1,t}=v'=v_h$.
Since $d_{u_1,v}=d_{u_1,v'}$, 
$v$ and $v'$ are the two Gromov product nodes for the shortest-path triangle $\Delta_{\left\{u_0,u_1,v_t\right\}}$
and thus applying Case~{\bf 1.1} on $\Delta_{\{u_0,u_1,v_t\}}$ we have 
$d_{v,v'}\leq 6\,\gadd_{ \Delta_{\{u_0,u_1,u_2\}} } +1$.
{\hfill{\Pisymbol{pzd}{113}}\vspace{0.1in}}

\begin{figure}
\epsfig{file=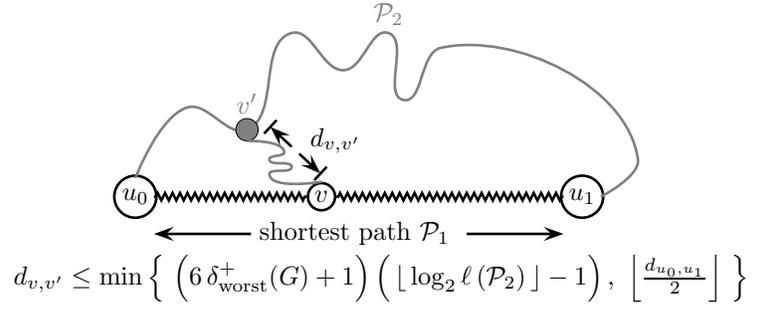,scale=0.9}
\caption{\label{f8}Illustration of the bound in Theorem~\ref{lem2}.}
\end{figure}
%
%

\section{Theorem~\ref{lem2} and Corollary~\ref{cor1}}
\label{lem2-sec}

\begin{theorem}[see \FI{f8} for a visual illustration]\label{lem2}
Let $\cP_1\equiv u_0\!\stackrel{\mathfrak{s}}{\leftrightsquigarrow}\! u_1$ 
and $\cP_2$ be a shortest path and an arbitrary path, respectively, between two nodes $u_0$ and $u_1$.
Then, for every node $v$ on $\cP_1$, there exists a node $v'$ on $\cP_2$ such that 
\[
\begin{array}{lcl}
d_{v,v'} & \leq  & 
\min \left\{ \,\, \Big( 6\,\gadd_{\mathrm{worst}}(G) + 2 \Big) \, \Big( \left\lfloor \, \log_2 \ell \left(\cP_2\right) \, \right\rfloor -1 \Big)\,,\,\, 
\left\lfloor \frac{d_{u_0,u_1}}{2} \right\rfloor \,\, \right\}
\\
[0.05in]
& = & { \mathrm{O} \Big( \, \gadd_{\mathrm{worst}}(G) \, \log \ell \left(\cP_2\right) \Big) }
\end{array}
\]
Since $\ell \left(\cP_2\right)\leq n$, the above bound also implies that
\[
d_{v,v'}\leq \Big( 6\,\gadd_{\mathrm{worst}}(G) + 2 \Big) \, \Big( \left\lfloor \, \log_2 n \, \right\rfloor -1 \Big)={ \mathrm{O}\Big( \, \gadd_{\mathrm{worst}}(G) \, \log n\Big) }
\]
\end{theorem}

\begin{corollary}\label{cor1}
Suppose that there exists a node $v$ on the shortest path between $u_0$ and $u_1$ such that $\min_{v'\in\cP_2} \left\{d_{v,v'} \right\}\geq\gamma$.
Then, $\ell \left( \cP_2 \right) \geq 2^{^{ \textstyle \frac{\gamma}{6\,\gadd_{\mathrm{worst}}(G) + 2} +1 } }-1=\Omega \Big( 2^{^{ \textstyle \gamma\,/\,\gadd_{\mathrm{worst}}(G) } }\Big)$.
\end{corollary}

\noindent
{\bf Proof of Theorem~\ref{lem2}.}
First, note that by selecting $v'$ to be one of $u_0$ or $u_1$ appropriately we have 
$
d_{v,v'}\leq \left\lfloor \nicefrac{d_{u_0,u_1}}{2} \right\rfloor
$.
Now, assume that $\ell \left( \cP_2 \right) >2$. 
Let $u_2$ be the node on the path $\cP_2$ such that 
$\ell \big( u_0 \!\stackrel{\cP_2}{\leftrightsquigarrow}\! u_2 \big)= \left\lceil\,\nicefrac{\ell\left(\cP_2\right)}{2}\,\right\rceil$. 
and consider the shortest-path triangle $\Delta_{\left\{u_0,u_1,u_2\right\}}$.
By Theorem~\ref{lem1} there exists a node $v'$ either on a shortest path between $u_0$ and $u_2$ or on a shortest path between $u_1$ and $u_2$ such that 
$d_{v,v'}\leq 6\,\gadd_{\mathrm{worst}}(G) + 2$. We move from $v$ to $v'$ and recursively solve the problem of finding a shortest path from $v'$ to a node on a part of the path $\cP_2$ containing 
at most $\left\lceil\nicefrac{\left(\cP_2\right)}{2}\right\rceil$ edges. 
Let $D(y)$ denote the minimum distance from $v$ to a node in a path of length $y$ between $u_0$ and $u_1$.
Thus, the worst-case recurrence for $D(y)$ is given by 
\[
\begin{array}{ll}
D(y) \leq D\left( \, \left\lceil\frac{y}{2}\right\rceil \, \right) + 6\,\gadd_{\mathrm{worst}}(G) + 2, & \mbox{if $y>2$}
\\[0.1in]
D(2) = 1  & 
\end{array}
\]
A solution to the above recurrence satisfies 
$
D \left(\ell \left(\cP_2 \right) \, \right)\leq \Big( 6\,\gadd_{\mathrm{worst}}(G) + 2 \Big) \, \Big( \left\lceil \, \log_2 \ell \left(\cP_2\right) \, \right\rceil -1 \Big)
$.
{\hfill{\Pisymbol{pzd}{113}}\vspace{0.1in}}



\section{Theorem~\ref{lem2-1} and Corollary~\ref{cor2-2}}
\label{lem2-1-sec}

\em
For easy of display of long mathematical equations, we will denote
$\gadd_{\mathrm{worst}}(G)$ simply as $\gadd$.

\begin{theorem}\label{lem2-1}
Let $\cP_1$ and $\cP_2$ be a shortest path and another path, respectively, between two nodes.
Define $\eta_{\cP_1,\cP_2}$ as 
{\small
\[
\begin{array}{l}
           \textstyle
	    \eta_{\,\cP_1,\cP_2}
	    \\
 \text{\footnotesize
	    $\textstyle =
           \Bigg( 6\,\gadd + 2 \Bigg) 
           \, 
           \log_2 \Bigg( \, 
           \bigg( 6\,\mu +2 \bigg) 
           \, 
           \bigg( 6\,\gadd + 2 \bigg) 
           \,
           \log_2 \bigg[ 
           \left( 6\,\gadd + 2 \right) 
           \, 
           \big( 3\,\mu +1 \big) 
           \,\mu \bigg] 
           \, +\mu 
           \Bigg)$
   }
\\
=
           \mathrm{O} \left(
           \gadd
           \,\log \left( \, 
               \mu\, \gadd
           \,\right)
           \, \right), 
           \text{ if $\cP_2$ is $\mu$-approximate short}
\\
\\
           \textstyle
           \eta_{\,\cP_1,\cP_2}
\\
	   =
           \textstyle
           \Bigg( 6\,\gadd + 2 \Bigg)  \, 
           \log_2 \Bigg( 
           8 \, \bigg( 6\,\gadd + 2 \bigg)  \, 
               \log_2 \bigg[ \left( 6\,\gadd + 2 \right)  
                               \, \left( 4+ 2\eps \right) \bigg] 
           \, + 1+\dfrac{\eps}{2} 
           \Bigg)
	   \\
           =
           \mathrm{O} \left(
           \gadd
           \log \Big( \, 
               \eps+
               \gadd \,
                 \log \eps 
           \,\Big)
           \, \right), 
           \text{ if $\cP_2$ is $\eps$-additive-approximate short}
\end{array}
\vspace*{5pt}
\]
}

\noindent
Then, the following statements are true.

\vspace*{0.1in}
\noindent
{\bf (a)} 
For every node $v$ on $\cP_1$, there exists a node $v'$ on $\cP_2$ such that 
$
d_{v,v'}\leq 
\left
\lfloor
\eta_{\cP_1,\cP_2}
\,
\right
\rfloor
$.

\vspace*{0.1in}
\noindent
{\bf (b)} 
For every node $v'$ on $\cP_2$, there exists a node $v$ on $\cP_1$ such that 
$d_{v,v'} \leq \zeta_{\cP_1,\cP_2}$ where 
{\small
\[
\zeta_{\cP_1,\cP_2}
=
\left\{
\begin{array}{l}
\,\,
\,\,
\min \left\{ \,\,
\left \lfloor \big( \, \mu + 1 \, \big) \,\eta_{\cP_1,\cP_2} + \dfrac {\mu}{2} \right\rfloor, 
\, \left\lfloor \dfrac{\mu\,d_{u_0,u_1} }{2} \right\rfloor \,\, \right\}
\\
[0.1in]
=
\mathrm{O}
\left( \mu \, 
           \gadd
           \,\log \left( \, 
               \mu\, \gadd
           \,\right)
\,\right), 
\text{ if $\cP_2$ is $\mu$-approximate short}
\\
[0.2in]
\,\,
\,\,
\min \left\{ \,
\left \lfloor 2\,\eta_{\cP_1,\cP_2} + \dfrac{1+\eps}{2} \right\rfloor,
\, \left\lfloor \dfrac{d_{u_0,u_1}+\eps }{2} \right\rfloor
\, \right\}
\\
[0.1in]
=
\mathrm{O} \left(
\eps + 
           \gadd
           \log \left( \, 
               \eps+
               \gadd \,
                 \log \eps 
           \,\right)
\,\right),
\\
[0.1in]
\hspace*{0.5in}
\text{ if $\cP_2$ is $\eps$-additive-approximate short}
\end{array}
\right.
\]
}
\end{theorem}

\begin{corollary}
\label{cor2-2}
{\bf (Hausdorff distance between approximate short paths)}
Suppose that $\cP_1$ and $\cP_2$ are two paths between two nodes. 
Then, 
the Hausdorff distance $d_H\left(\cP_1,\cP_2\right)$ between these two paths can be bounded as follows: 
\begin{multline*}
\hspace*{-0.1in}
d_H\left(\cP_1,\cP_2\right)
\eqdef
\max 
\left\{ \,
\max_{v_1\,\in \,\cP_1} \min_{v_2\,\in \,\cP_2} \Big\{\, d_{v_1,v_2} \,\Big\}
, \, 
\max_{v_2\,\in \,\cP_2} \min_{v_1\,\in \,\cP_1} \Big\{\, d_{v_1,v_2} \,\Big\}
\, \right\}
\\
\leq
\eta_{\cP_1, \textstyle u_0\stackrel{\mathfrak{s}}{\leftrightsquigarrow} u_1}
+ \zeta_{\cP_2, \textstyle u_0\stackrel{\mathfrak{s}}{\leftrightsquigarrow} u_1}
\end{multline*}
\end{corollary}

\begin{corollary}\label{cor2-1}
Suppose that there exists a node $v$ on the shortest path between $u_0$ and $u_1$ such that $\min_{\,v'\,\in\,\cP_2} \left\{d_{v,v'} \right\}\geq\gamma$.
Then, the following is true.

\vspace*{0.1in}
\noindent
$\bullet$
If $\cP_2$ is a $\mu$-approximate short path then 
{
\[
           \mu
      > 
\frac
{
   2^{
   \textstyle
   \frac{ \gamma }
   {
   6\,\gadd + 1
   }
   }
}
{
   { 12\,\gamma }
   - 
              \Big( \, 24 + \mathrm{o} (1) \, \Big)
              \,
              \Big( \, 6\,\gadd + 1 \, \Big)
}
- \frac { 1 } { 3 }
\,\Rightarrow\,
\mu
=
\Omega
\left(
\frac{
2^{ \textstyle
\nicefrac{\gamma}{\gadd}
}
}
{
\gamma
}
\right)
\]
}
$\bullet$
If $\cP_2$ is a $\eps$-additive-approximate short path then 
\[
\begin{array}{l}
\eps
\,>\, 
\dfrac{
2^{
\frac{
\gamma
}
{
6\,\gadd + 1
}
}
}
{
\Big( 48\,\gadd \,+ \frac{17}{2} \Big) 
}
\,-\,
\log_2 \left( 48\,\gadd + 8 \right) 
\\
\hspace*{0.4in}
\Rightarrow\,
\eps
=
\Omega
\left(
\dfrac{
2^{ \textstyle
\nicefrac{\gamma}{\gadd}
}
}
{
\gadd
}
\,-\,
\log \gadd
\right)
\end{array}
\]
In particular, assuming real world networks have small constant values of 
$\gadd$, 
the asymptotic dependence of $\mu$ and $\eps$ on $\gamma$ can be summarized as:
\[
\text{both $\mu$ and $\eps$ are } 
\Omega
\left(
2^{ 
\,c \: \gamma
}
\,
\right)
\text{ for some constant $0<c<1$} 
\]
\end{corollary}

\vspace*{0.1in}
\noindent
{\bf Proof of Theorem~\ref{lem2-1}.}
Let $\cP_1$ and $\cP_2$ be a shortest path and another path, respectively, between two nodes $u_0$ and $u_1$.
Note that any ``sub-path'' of a $\mu$-approximate short path is {\em also} a $\mu$-approximately short path, \IE, 
$u_i\!\stackrel{\cP}{\leftrightsquigarrow}\! u_j$ is also a $\mu$-approximate short path, and similarly 
any sub-path of a $\eps$-additive-approximate short path is {\em also} a $\eps$-additive-approximate short path.
$\mu$-approximate shortest paths also restrict the ``span'' of a path-chord of the path, \IE, 
{\em if $\big(u_0,u_1,\dots,u_k\big)$ is a $\mu$-approximate short path and $\left\{u_i,u_j\right\}\in E$ then $|j-i\,|\leq\mu$}.

\vspace*{0.1in}
\noindent
{\bf (a)}
Let $v$ and $v'$ be two nodes on $\cP_1$ and $\cP_2$, respectively, such that $\displaystyle\alpha=d_{v,v'}=\max_{v''\in\cP_1} \min_{v'''\in\cP_2} \left\{ d_{v'',v'''} \right\}$.
Let $v_\ell \in u_0\!\stackrel{\cP_1}{\leftrightsquigarrow}\! v$ and $v_r \in u_1\!\stackrel{\cP_1}{\leftrightsquigarrow}\! v$ be two nodes defined by 
\begin{gather*}
d_{v_\ell,v}=
\begin{array}{ll}
2\,\alpha+1, & \mbox{if $d_{u_0,v}>2\,\alpha+1$} 
\\
d_{u_0,v}, & \mbox{otherwise}
\end{array}
\\
d_{v_r,v}=
\begin{array}{ll}
2\,\alpha+1, & \mbox{if $d_{u_1,v}>2\,\alpha+1$} 
\\
d_{u_1,v}, & \mbox{otherwise}
\end{array}
\end{gather*}
By definition of $\alpha$, there exists two nodes 
$\widetilde{v_\ell}$ and $\widetilde{v_r}$ on the path $\cP_2$ such that 
$d_{v_\ell,\widetilde{v_\ell}},d_{v_r,\widetilde{v_r}}\leq\alpha$.
Consider the $\cP_3= \widetilde{v_\ell} \!\stackrel{\cP_2}{\leftrightsquigarrow}\!  \widetilde{v_r}$ that is the part of path $\cP_2$ from  
$\widetilde{v_\ell}$ to $\widetilde{v_r}$.
Note that 
\[
d_{\widetilde{v_\ell},\widetilde{v_r}}
\leq
d_{\widetilde{v_\ell},{v_\ell}}
+
d_{{v_\ell},{v_r}}
+
d_{{v_r},\widetilde{v_r}}
\leq
6\,\alpha+2
\]
Thus, we arrive at the following inequalities
\[
\ell \left( \cP_3 \right) \leq 
\begin{array}{ll}
\big( 6\,\alpha+2 \big)\,\mu, & \mbox{if $\cP_2$ is $\mu$-approximate short}
\\
[0.03in]
6\,\alpha+2 +\eps, & \mbox{if $\cP_2$ is $\eps$-additive-approximate short}
\end{array}
\]
Now consider the path 
$\cP_4=
v_\ell 
\!\stackrel{\mathfrak{s}}{\leftrightsquigarrow}\! 
\widetilde{v_\ell}
\!\stackrel{\cP_2}{\leftrightsquigarrow}\! 
\widetilde{v_r}
\!\stackrel{\mathfrak{s}}{\leftrightsquigarrow}\! 
v_r 
$
obtained by taking a shortest path from 
$v_\ell$ to 
$\widetilde{v_\ell}$
followed by the path $\cP_3$ followed by 
a shortest path from 
$v_r$ to 
$\widetilde{v_r}$.
Note that 
\[
\ell \left( \cP_4 \right) \leq 
\left\{
\begin{array}{l}
\big( 6\,\alpha+2 \big)\,\mu + 2\,\alpha, 
\mbox{ if $\cP_2$ is $\mu$-approximate short}
\\
[0.07in]
6\,\alpha+2 +\eps + 2\,\alpha=8\,\alpha+2+\eps, 
\\
\hspace*{0.4in}
\mbox{if $\cP_2$ is $\eps$-additive-approximate short}
\end{array}
\right.
\]
We claim that $\min_{\,\widetilde{v}\,\in\,\cP_4} \{ d_{v,\widetilde{v}} \} =\alpha$.
Indeed, if $\widetilde{v}\in\cP_3$ then, by definition of $\alpha$,  
$\min_{\,\widetilde{v}} \, \{ d_{v,\widetilde{v}} \} =\alpha$. 
Otherwise, if
$\widetilde{v}\in
v_\ell 
\!\stackrel{\mathfrak{s}}{\leftrightsquigarrow}\! 
\widetilde{v_\ell}
$, then by triangle inequality   
$
d_{v_\ell,v} \leq d_{v,\widetilde{v}} + d_{\,\widetilde{v},v_\ell}
\, \Rightarrow \,
d_{v,\widetilde{v}} \geq 
2\,\alpha + 1 - d_{\,\widetilde{v},v_\ell} > \alpha
$.
Similarly, if
$\widetilde{v}\in
\widetilde{v_r}
\!\stackrel{\mathfrak{s}}{\leftrightsquigarrow}\! 
v_r 
$, then by triangle inequality   
$
d_{v_r,v} \leq d_{v,\widetilde{v}} + d_{\,\widetilde{v},v_r}
\, \Rightarrow \,
d_{v,\widetilde{v}} \geq 
2\,\alpha + 1 - d_{\,\widetilde{v},v_r} > \alpha
$.
Since 
$
v_\ell
\!\stackrel{\cP_1}{\leftrightsquigarrow}\! 
v_r
$
is a shortest path between $v_\ell$ and $v_r$ and $v$ is a node on this path,
by Theorem~\ref{lem2}, 
$
\alpha 
\leq
\left( 6\,\gadd + 2 \right) \, \left( \,\left\lfloor \, \log_2 \ell \left(\cP_4\right) \, \right\rfloor -1 \right)
$.
Thus, we have the following inequalities:

\vspace*{0.05in}
\noindent
$\bullet$ 
If $\cP_2$ is a $\mu$-approximate short path then 
\begin{gather}
\begin{array}{c l}
&
\ell \left(\cP_4 \right)
\\
\leq 
&
\big( 6\,\alpha+2 \big)\,\mu + 2\,\alpha
\\
[0.03in]
=
&
\big( 6\,\mu+2 \big)\,\alpha + 2\,\mu
\\
[0.03in]
\leq
&
\big( 6\,\mu+2 \big)\,
\left( 6\,\gadd + 2 \right) \, \left( \log_2 \ell \left(\cP_4\right) -1 \right)
+ 2\,\mu
\\
[0.03in]
\leq
&
\big( 6\,\mu+2 \big)\,
\left( 6\,\gadd + 2 \right) \, \left( \log_2 
\left( \big( 6\,\mu+2 \big)\,\alpha + 2\,\mu \right) \,
-1 \right)
+ 2\,\mu
\\
[0.03in]
\Rightarrow \,
&
\alpha
\leq 
\left( 6\,\gadd + 2 \right) \, \left( \log_2 
\left( \big( 3\,\mu+1 \big)\,\alpha + \mu \right) \,
\right)
\label{loglog1}
\end{array}
\end{gather}

\vspace*{0.05in}
\noindent
$\bullet$ 
If $\cP_2$ is a $\eps$-additive-approximate short path then 
\begin{gather}
\begin{array}{r c l}
\ell \left(\cP_4 \right)
&
\leq 
&
8\,\alpha + 2 + \eps
\\
[0.03in]
&
\leq
&
8\,
\left( 6\,\gadd + 2 \right) \, \left( \log_2 \ell \left(\cP_4\right) -1 \right)
+ 2 + \eps
\\
[0.03in]
&
\leq
&
8 \,
\left( 6\,\gadd + 2 \right) \, \left( \log_2 
\left( 8\,\alpha + 2 + \eps \right) \,
-1 \right)
+ 2 + \eps
\\
[0.03in]
&
\Rightarrow \,
&
8\,\alpha + 2 + \eps
\\
&
&
\,\,\,
\leq
8 \,
\left( 6\,\gadd + 2 \right) \, \left( \log_2 
\left( 8\,\alpha + 2 + \eps \right) \,
-1 \right)
+ 2 + \eps
\\
[0.03in]
&
\equiv\,
&
\alpha \leq 
\left( 6\,\gadd + 2 \right) \, \left( \log_2 
\left( 4\,\alpha + 1 + \frac{\eps}{2} \right) \,
\right)
\label{loglog2}
\end{array}
\end{gather}
Both ~\eqref{loglog1} and ~\eqref{loglog2} are of the form 
$
\alpha \leq a \log_2 \big( b\,\alpha + c \big)
\,\equiv\,
2^{ \nicefrac{\alpha}{a} }
\leq b\,\alpha+c
$
where 
\[
\begin{array}{lcl}
a & =  & 6\,\gadd + 2 \geq 1 
\text{ for both \eqref{loglog1} and \eqref{loglog2}}
\\
[-0.15in]
\end{array}
\]
\[
\begin{array}{lcllcl}
b & = & 
\left\{
\begin{array}{ll}
3\,\mu+1 \geq 4  & \text{for \eqref{loglog1}}
\\
4 & \text{for \eqref{loglog2}}
\end{array}
\right.
&
c & = & 
\left\{
\begin{array}{ll}
\mu \geq 1 & \text{for \eqref{loglog1}}
\\
1+\frac{\eps}{2} \geq 1 & \text{for \eqref{loglog2}}
\end{array}
\right.
\end{array}
\]
Thus, $\alpha$ is at most $z_0$ where $z_0$ is the largest positive integer value of $z$ that satisfies the equation:
\begin{gather*}
2^{ \nicefrac{z}{a} }
\leq b\,z+c
\end{gather*}
In the sequel, we will use the fact that $\log_2 \big(x\,y+1 \big) \geq \log_2 \big( x+y \big)$ for $x,y\geq 1$. This holds since
\[
x \geq 1 \,\,\&\,\, y \geq 1 
\, \Rightarrow \,
y\,(x-1) \geq x-1
\,\equiv\,
x\,y+1\geq x+y
\]
We claim that 
$z_0\leq \eta = a \, \log_2 \left( 2 \, a \, b \log_2 \big(a \, b\,c \big) \, +c \right)$.
This is verified by showing that $2^{ \nicefrac{\eta}{a} }\geq b\,\eta+c$ as follows:
\begin{gather*}
\begin{array}{r c l}
2^{\nicefrac{\eta}{a}} 
&
=
&
2^{ 
\log_2 \big( 2 \, a \, b \log_2 \big(a \, b\, c\big) \,  +c \big)
}
=2\,a\,b\,\log_2 \big( a \, b \, c \big) \,  + c
\\
b\,\eta + c 
&
= 
&
a \, b\, 
\left( 
\log_2 \big( 2 \, a \, b \log_2 \big(a \, b \, c \big) \, +c \big)
\, \right)
+ c
\end{array}
\end{gather*}
\vspace*{-0.15in}
\begin{multline*}
\begin{array}{r l}
& 
2^{\nicefrac{\eta}{a}} 
>
b\,\eta + c
\\
\equiv
& 
2\,a\,b\,\log_2 \big( a \, b \, c \big) \,  + c
\geq
a \, b\, 
\left( 
\log_2 \big( 2 \, a \, b \log_2 \big(a \, b \, c \big) \, +c \big)
\, \right)
+ c
\\
\equiv
&
2\,\log_2 \big( a \, b \, c \big) 
\, \geq \,
\log_2 \big( 2 \, a \, b \log_2 \big(a \, b \, c \big) \, +c \big)
\\
\Leftarrow
&
2\,\log_2 \big( a \, b \, c \big)
\, \geq \,
\log_2 \left( 2 \, a \, b\,c \log_2 \big(a \, b \, c \big) + 1 \right)
\\
& 
\,\,\,\,
\,\,\,\,
\text{\small since $2 \, a \, b \log_2 \big(a \, b \, c \big)\geq 1$ and $c\geq 1$}
\\
\equiv
&
\big( a \, b \, c \big)^2 
\, \geq \,
 2 \, a \, b\,c \log_2 \big(a \, b \, c \big) + 1
\\
\Leftarrow
& 
a \, b \, c 
\, \geq \,
\log_2 \big(a \, b \, c \big) + 1
\end{array}
\end{multline*}
and the very last inequality holds since $a\,b\,c\geq 4$.
Thus, we arrive at the at the following bounds:

\vspace*{0.1in}
\noindent
$\bullet$ 
If $\cP_2$ is a $\mu$-approximate short path then 
{\footnotesize
\[
\eta
=
\Bigg( 6\,\gadd + 2 \Bigg) 
\, 
\log_2 \Bigg( \, 
\bigg( 6\,\mu +2 \bigg) 
\, 
\bigg( 6\,\gadd + 2 \bigg) 
\,
\log_2 \bigg[
\left( 6\,\gadd + 2 \right) 
\, 
\big( 3\,\mu +1 \big) 
\,\mu \bigg] \, +\mu \Bigg)
\]
}
$\bullet$ 
If $\cP_2$ is a $\eps$-additive-approximate short path then 
{\footnotesize
\[
\eta
=
\Bigg( 6\,\gadd + 2 \Bigg)  \, 
\log_2 \Bigg(
8 \, \bigg( 6\,\gadd + 2 \bigg)  \, 
    \log_2 \bigg[ \left( 6\,\gadd + 2 \right)  
                    \, \left( 4+ 2\eps \right) \, \bigg] 
\, + 1+\frac{\eps}{2} 
\,\Bigg)
\]
}

\noindent
{\bf (b)}
Let the ordered sequence of nodes in the path $\cP_3=v_1\!\stackrel{\cP_2}{\leftrightsquigarrow}\! v_1'$ be a (length) {\em maximal} sequence of nodes such 
that: 
\[
\forall \, v' \in \cP_3 \colon 
\min_{v \,\in \,\cP_1} \left\{ \, d_{v,v'} \right\}  > Z_{\cP_1,\cP_2}
\]
Consider the following set of nodes belonging to the two paths 
$u_0\!\stackrel{\cP_2}{\leftrightsquigarrow}\! v_1$ and $v_1'\!\stackrel{\cP_2}{\leftrightsquigarrow}\! u_1$: 
\[
\begin{array}{l}
\cS_\ell=
\bigcup
\left\{
v' \! \in u_0 \! \stackrel{\cP_2}{\leftrightsquigarrow} \! v_1 \, \big| \, \exists \, v\in \cP_1 \colon d_{v,v'}= \! \min_{v'' \,\in \,\cP_2} \left\{ \, d_{v,v''} \right\}  
\right\}
\\
\cS_r=
\bigcup
\left\{
v' \! \in v_1' \! \stackrel{\cP_2}{\leftrightsquigarrow} \! u_1
\, \big| \, \exists \, v\in \cP_1 \colon d_{v,v'}= \! \min_{v'' \,\in \,\cP_2} \left\{ \, d_{v,v''} \right\}  
\right\}
\end{array}
\]
Since $u_0\in \cS_\ell$ and $u_1\in \cS_r$, it follows that $\cS_\ell\neq\emptyset$ and $\cS_r\neq\emptyset$. 
Note that 
\[
\hspace*{-0.0in}
\bigcup \left\{ v \in 
u_0 \! \stackrel{\cP_1}{\leftrightsquigarrow} \! u_1 \, \big| \,
\exists \, v' \in \cS_\ell \cup \cS_r \colon 
d_{v,v'}= \! \min_{v'' \,\in \,\cP_2} \left\{ \, d_{v,v''} \right\}
\right\} 
=
\!\!\!\!\!
\bigcup_{v \, \in \, u_0  \, \stackrel{\cP_1}{\leftrightsquigarrow}  \, u_1 } \!\!\!\!\! \Big\{ \, v\,  \Big\}
\]
Thus, there exists two adjacent nodes $v_4$ and $v_4'$ on $\cP_1$ such that both $d_{v_4,v_3}$ and $d_{v_4',v_3'}$ is at most $Z_{\cP_1,\cP_2}$.
Using triangle inequality it follows that 
\[
d_{v_3,v_3'} \leq d_{v_3,v_4} + d_{v_4,v_4'} + d_{v_4',v_3'} = 2\,Z_{\cP_1,\cP_2} + 1
\]
giving the following bounds 
\[
\ell \left( v_3 \! \stackrel{\cP_2}{\leftrightsquigarrow} \! v_3' \right) \leq 
\left\{
\begin{array}{l}
\mu\, d_{v_3,v_3'} \leq 2\,\mu\,Z_{\cP_1,\cP_2} + \mu, 
\\
   \hspace*{0.5in}
   \text{if $\cP_2$ is $\mu$-approximate short}
\\
[0.05in]
d_{v_3,v_3'} +\eps \leq 2\,Z_{\cP_1,\cP_2} + 1 + \eps,
\\
   \hspace*{0.5in}
   \text{if $\cP_2$ is $\eps$-additive-approximate short}
\end{array}
\right.
\]
For any node $v'$ on $\cP_3$, we can always use the following path to reach a node on $\cP_1$:
\begin{itemize}
\item
if $d_{v',v_3}\leq d_{v',v_3'}$ then we take the path 
$
v'\!\stackrel{\cP_2}{\leftrightsquigarrow}\! v_3
\!\stackrel{\mathfrak{s}}{\leftrightsquigarrow}\! v_4
$
of length
at most 
$\bigg\lfloor\frac{\ell \left( v_3 \stackrel{\cP_2}{\leftrightsquigarrow} v_3' \right)}{2}\bigg\rfloor
+ Z_{\cP_1,\cP_2}
$
to reach the node $v=v_4$ on $\cP_1$; 

\item
otherwise
we take the path 
$
v'\!\stackrel{\cP_2}{\leftrightsquigarrow}\! v_3'
\!\stackrel{\mathfrak{s}}{\leftrightsquigarrow}\! v_4'
$
of length
at most 
$\bigg\lfloor\frac{\ell \, \left( v_3 \stackrel{\cP_2}{\leftrightsquigarrow} v_3' \right)}{2}\bigg\rfloor
+ Z_{\cP_1,\cP_2}
$
to reach the node $v=v_4'$ on $\cP_1$.
\end{itemize}
This gives the following worst-case bounds for $d_{v,v'}$:
\[
d_{v,v'} \leq 
\left\{
\begin{array}{l}
\left \lfloor \big( \, \mu + 1 \, \big) \,Z_{\cP_1,\cP_2} + \frac {\mu}{2} \right\rfloor,
   \text{if $\cP_2$ is $\mu$-approximate short}
\\
[0.05in]
\left \lfloor 2\,Z_{\cP_1,\cP_2} + \frac{1+\eps}{2} \right\rfloor,
   \text{if $\cP_2$ is $\eps$-additive-approximate short}
\end{array}
\right.
\]
{\hfill{\Pisymbol{pzd}{113}}\vspace{0.1in}}

\begin{figure}
\epsfig{file=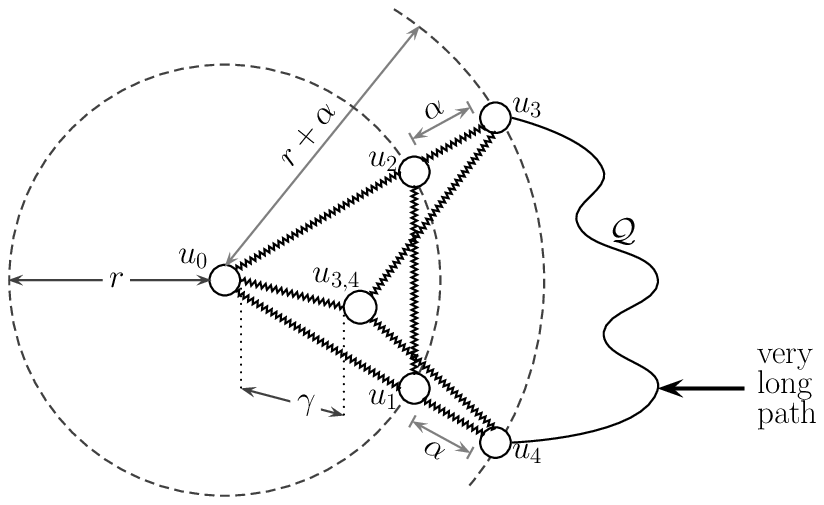}
\caption{\label{f11}Illustration of the claims in Theorem~\ref{lem3} and Corollary~\ref{cor2}.}
\end{figure}
%

\section{Theorem~\ref{lem3} and Corollary~\ref{cor2}}
\label{lem3-sec}

\begin{theorem}[see \FI{f11} for a visual illustration]
\label{lem3}
Suppose that we are given the following:

\noindent
$\bullet$
three integers $\kappa\geq 4$, $\alpha>0$, $r>\left(\frac{\kappa}{2}-1\right)\left( 6\,\gadd_{\mathrm{worst}}(G) + 2 \right)$,

\noindent
$\bullet$
five nodes $u_0,u_1,u_2,u_3,u_4$ such that
\begin{itemize}
\item
$u_1,u_2\in B_r\left(u_0\right)$ with $d_{u_1,u_2}\geq \frac{\kappa}{2} \, \left( 6\,\gadd_{\mathrm{worst}}(G) + 2 \right)$,

\item
$d_{u_1,u_4}=d_{u_2,u_3}=\alpha$. 
\end{itemize}
Then, the following statements are true for any shortest path $\cP$ between $u_3$ and $u_4$:

\vspace*{0.05in}
\noindent
{\bf ({\em a})}
there exists a node $v$ on $\cP$ such that 
\[
d_{u_0,v}\leq r-\left(\frac{3\kappa-2}{12} \right) \left( 6\,\gadd_{\mathrm{worst}}(G) + 2 \right)
=
r - 
\mathrm{O} \Big( \kappa \, \gadd_{\mathrm{worst}}(G) \Big)
\]

\noindent
{\bf ({\em b})}
$\ell \left(\cP \right)\geq \left( \frac{3 \kappa-2}{6} \right) \left( 6\,\gadd_{\mathrm{worst}}(G) + 2 \right) +2\,\alpha 
= 
{\Omega \, \Big( \kappa \,\gadd_{\mathrm{worst}}(G) \, +\,  \alpha \Big)}$.
\end{theorem}

\begin{corollary}[see \FI{f11} for a visual illustration]
\label{cor2}
Consider any path $\cQ$ between $u_3$ and $u_4$ that does not involve a node in $\cup_{r'\leq r}\mathcal{B}_{r'}\left(u_0\right)$. 
Then, the following statements hold:
\begin{description}
\item[(i)]
$\ell \left( \cQ \right) \geq 
2^{^{
\textstyle
\frac{\alpha} { 6\,\gadd_{\mathrm{worst}}(G) + 2 }
+
\frac{\kappa}{4} + \frac{5}{6}
}
}
-1
=
2^{\textstyle \Omega \, \left( \frac{\alpha}{\gadd_{\mathrm{worst}}(G)} \,+\, \kappa \right)}
$.
In particular, if $\gadd_{\mathrm{worst}}(G)$ is a constant then 
$\ell \left( \cQ \right) = {2^{\textstyle \Omega \left( \alpha + \kappa \right)}}$
and thus 
$\ell \left( \cQ \right)$ increases at least exponentially with both $\alpha$ and $\kappa$.

\vspace*{0.1in}
\item[(ii)]
if $\cQ$ is a $\mu$-approximate short path then 
{\small
\[
\mu \, \geq \,
\frac
{
   2^{
   \textstyle
\frac{\alpha} { 6\,\gadd_{\mathrm{worst}}(G) + 2 }
+
\frac{\kappa}{4} - \frac{1}{6}
   }
}
{
   { 12\,\alpha + \big( 3\,\kappa -26 - \mathrm{o} (1) \, \big)\, \big(6\,\gadd_{\mathrm{worst}}(G) + 2 \big)   }
}
- \frac { 1 } { 3 }
=
\Omega
\left(
\frac{
2^{
\Theta \, \left( \frac{\alpha}{\gadd_{\mathrm{worst}}(G)} \,+\, \kappa\right)
}
}
{
\alpha
+ \kappa\,
\gadd_{\mathrm{worst}}(G)
}
\right)
\]
}
In particular, if $\gadd_{\mathrm{worst}}(G)$ is a constant then 
$\mu = {\Omega \left( \frac{ 2^{\,\Theta \,( \alpha+\kappa )}  } { \alpha + \kappa } \right)}$
and thus 
$\mu$ increases at least exponentially with both $\alpha$ and $\kappa$.

\vspace*{0.1in}
\item[(iii)]
if $\cQ$ is a $\eps$-additive-approximate short path then 
\[
\eps > 
\frac{
2^{
   \textstyle
\frac{\alpha} { 6\,\gadd_{\mathrm{worst}}(G) + 2 }
+
\frac{\kappa}{4} - \frac{1}{6}
}
}
{
48\,\gadd_{\mathrm{worst}}(G) \,+ \frac{17}{2}
}
\,-\,
\log_2 \left( 48\,\gadd_{\mathrm{worst}}(G) + 16 \right) 
\]
In particular, if $\gadd_{\mathrm{worst}}(G)$ is a constant then 
$\eps = {\Omega \left( 2^{\,\Theta \,( \alpha+\kappa )} \right)}$
and thus 
$\eps$ increases at least exponentially with both $\alpha$ and $\kappa$.
\end{description}
\end{corollary}

\subsection{Proof of Theorem~\ref{lem3}}

Consider the shortest-path triangle $\Delta_{\left\{u_0,u_3,u_4\right\}}$ and 
let $u_{0,3},u_{0,4}$ and $u_{3,4}$ be 
the Gromov product nodes of $\Delta_{\{u_0,u_3,u_4\}}$ on the sides (shortest paths) $u_0$ to $u_3$, $u_0$ to $u_4$ and $u_3$ to $u_4$, respectively.
Thus, $d_{u_0,u_{0,3}}=d_{u_0,u_{0,4}}$, and $\beta=d_{u_3,u_{3,4}}=\left\lfloor \frac{d_{u_0,u_3}+d_{u_3,u_4}-d_{u_0,u_4}}{2} \right\rfloor = \left\lfloor \frac{d_{u_3,u_4}}{2} \right\rfloor$
since $d_{u_0,u_3}=d_{u_0,u_4}=r+\alpha$.

We first claim that $d_{u_0,u_{0,3}}<r=d_{u_0,u_2}$. Suppose for the sake of contradiction that $d_{u_0,u_{0,3}}=d_{u_0,u_{0,4}}\geq r$.
Then, by Theorem~\ref{lem1} we get $d_{u_1,u_2}\leq 6\,\gadd_{\mathrm{worst}}(G) + 2$ which contradicts the assumption that 
$d_{u_1,u_2}\geq \frac{\kappa}{2} \, \left( 6\,\gadd_{\mathrm{worst}}(G) + 2 \right)$ since $\kappa\geq 4$.

Thus, assume that $d_{u_0,u_{0,3}}=d_{u_0,u_{0,4}}=r-x$ for some integer $x>0$.  By Theorem~\ref{lem1}, $d_{u_{0,3},u_{0,4}}\leq 6\,\gadd_{\mathrm{worst}}(G) + 2$. 
Let $d_{u_{0,3},u_{0,4}}=6\,\gadd_{\mathrm{worst}}(G) + 2-y$ for some integer $0<y\leq 6\,\gadd_{\mathrm{worst}}(G) + 2$ and 
$d_{u_1,u_2}=\frac{\kappa}{2} \, \left( 6\,\gadd_{\mathrm{worst}}(G) + 2 \right)+z$ for some integer $z\geq 0$.
Consider the $4$-node condition for the four nodes $u_1,u_2,u_{0,3},u_{0,4}$.
The three relevant quantities for comparison are: 
\[
\begin{array}{l}
q_{\parallel}=d_{u_1,u_2}+d_{u_{0,3},u_{0,4}}= \left(\frac{\kappa}{2}+1\right) \left( 6\,\gadd_{\mathrm{worst}}(G) + 1 \right)+z -y
\\
[0.05in]
\hspace*{-0.1in}
q_{=}=d_{u_{0,3},u_2}+d_{u_{0,4},u_1}=
\left(d_{u_0,u_2}-d_{u_0,u_{0,3}}\right) + \left(d_{u_0,u_1}-d_{u_0,u_{0,4}}\right) =2x
\\
[0.05in]
q_{\varparallelinv}=d_{u_{0,3},u_1}+d_{u_{0,4},u_2}
\leq
\left( d_{u_{0,3},u_{0,4}}+d_{u_{0,4},u_1} \right) + \left( d_{u_{0,3},u_{0,4}}+d_{u_{0,3},u_2} \right)
\\
\hspace*{1.1in}
=
12\,\gadd_{\mathrm{worst}}(G) + 4-2y+2x
\end{array}
\]
We now show that $x> \left(\frac{3\kappa-2}{12}\right) \left( 6\,\gadd_{\mathrm{worst}}(G) + 2 \right)$. 
We have the following cases.
\begin{itemize}
\item
Assume that $q_{\varparallelinv} \leq \min \left\{ q_{\parallel},\, q_{=} \right\}$. This implies
\[
\begin{array}{cl}
&
\big| q_{\parallel}-q_{=} \big| \leq 2\,\gadd_{\mathrm{worst}}(G) 
\\
[0.05in]
\equiv
&
\Big| \, \left( \frac{\kappa}{2} +1\right) \,\left( 6\,\gadd_{\mathrm{worst}}(G) + 2 \right)+z -y-2x \, \Big| \leq 2\,\gadd_{\mathrm{worst}}(G)
\\
\Rightarrow
&
x \geq \dfrac{ \left( \frac{\kappa}{2} +1\right) \, \left( 6\,\gadd_{\mathrm{worst}}(G) + 2 \right)+z -y - 2\,\gadd_{\mathrm{worst}}(G)}{2}
\\
&
\hspace*{0.1in}
\geq
\left(\frac{3\kappa-2}{12}\right) \left( 6\,\gadd_{\mathrm{worst}}(G) + 2 \right)
\,+\, \frac{1}{6}
\end{array}
\]

\item
Otherwise, assume that $q_{=} \leq \min \left\{ q_{\parallel},\, q_{\varparallelinv} \right\}$. 
This implies
\[
\hspace*{-0.0in}
\begin{array}{cl}
&
\big| q_{\parallel}-q_{\varparallelinv} \big| \leq 2\,\gadd_{\mathrm{worst}}(G) 
\\
[0.05in]
\Rightarrow
&
q_{\varparallelinv} \geq  q_{\parallel} - 2\,\gadd_{\mathrm{worst}}(G) 
\\
[0.05in]
\Rightarrow
&
d_{u_{0,3},u_1}+d_{u_{0,4},u_2}
\geq  
\left( \frac{\kappa}{2} +1\right) \, \left( 6\,\gadd_{\mathrm{worst}}(G) + 2 \right)+z -y - 2\,\gadd_{\mathrm{worst}}(G) 
\\
[0.05in]
\Rightarrow
&
\big( d_{u_{0,3},u_{0,4}} + d_{u_{0,4},u_1} \big) + \big( d_{u_{0,3},u_{0,4}} + d_{u_{0,3},u_2} \big) 
\geq 
d_{u_{0,3},u_1}+d_{u_{0,4},u_2}
\\
&
\hspace*{0.4in}
\geq 
\left( \frac{\kappa}{2} +1\right) \, \left( 6\,\gadd_{\mathrm{worst}}(G) + 2 \right)+z -y - 2\,\gadd_{\mathrm{worst}}(G) 
\\
[0.05in]
\Rightarrow
&
2x + 2 \,\Big( 6\,\gadd_{\mathrm{worst}}(G) + 2-y \, \Big) 
\\
&
\hspace*{0.4in}
\geq 
\left( \frac{\kappa}{2} +1\right) \, \left( 6\,\gadd_{\mathrm{worst}}(G) + 2 \right)+z -y - 2\,\gadd_{\mathrm{worst}}(G) 
\\
[0.05in]
\Rightarrow
&
x \geq 
\left(\frac{3\kappa-2}{12}\right) \left( 6\,\gadd_{\mathrm{worst}}(G) + 2 \right)
\,+\, \frac{1}{6}
\end{array}
\]

\item 
Otherwise, assume that $q_{\parallel} \leq \min \left\{ q_{=},\, q_{\varparallelinv} \right\}$. This implies
\[
\begin{array}{cl}
&
\big| q_{=}-q_{\varparallelinv} \big| \leq 2\,\gadd_{\mathrm{worst}}(G) 
\\
[0.05in]
\equiv
&
\big|\,  2x - \big( d_{u_{0,3},u_1}+d_{u_{0,4},u_2} \big) \, \big| \leq 2\,\gadd_{\mathrm{worst}}(G) 
\\
[0.05in]
\Rightarrow
&
2x \geq d_{u_{0,3},u_1}+d_{u_{0,4},u_2} - 2\,\gadd_{\mathrm{worst}}(G) 
\\
&
\hspace*{0.4in}
\geq 
\big( d_{u_1,u_2} - d_{u_{0,4},u_1}  \big) + \big( d_{u_1,u_2} - d_{u_{0,3},u_1} \big) - 2\,\gadd_{\mathrm{worst}}(G) 
\\
[0.05in]
\equiv
&
2x \geq 
\kappa \left( 6\,\gadd_{\mathrm{worst}}(G) + 2 \right)+2z - 2x - 2\,\gadd_{\mathrm{worst}}(G)
\\
[0.05in]
\Rightarrow
&
x \geq \left( \frac{3\kappa-2}{12} \right) \left( 6\,\gadd_{\mathrm{worst}}(G) + 2 \right)
+ \frac{\gadd_{\mathrm{worst}}(G) }{2}
+ \frac{1}{6}
\end{array}
\]
\end{itemize}
Using Theorem~\ref{lem1}, it now follows that 
\begin{multline*}
d_{u_0,u_{3,4}} \leq d_{u_0,u_{0,3}} + d_{u_{0,3},u_{0,4}}
\leq \big(r-x\big) + \left( 6\,\gadd_{\mathrm{worst}}(G) + 2 \right)
\\
< 
r-\left(\frac{3\kappa-2}{12} \right) \left( 6\,\gadd_{\mathrm{worst}}(G) + 2 \right)
\end{multline*}
This proves part {\bf ({\em a})} with $u_{3,4}$ being the node in question. To prove part {\bf ({\em b})}, note that 
\[
|\cP|=2\,\beta \geq  2(r+\alpha)-2d_{u_0,u_{3,4}}
\geq 
2\alpha + 
\left( \dfrac{3 \kappa-2}{6} \right) \left( 6\,\gadd_{\mathrm{worst}}(G) + 2 \right)
\]



\subsection{Proof of Corollary~\ref{cor2}}

Consider such a path $\cQ$ and consider the node $u_{3,4}$ on the shortest path between $u_3$ and $u_4$.
Since every node of $\cQ$ is at a distance strictly larger than $r+\alpha$ from $u_0$, 
by Theorem~\ref{lem3} the following holds for every node $v\in\cQ$ 
\begin{multline*}
\hspace*{-0.2in}
d_{u_{3,4},v} \geq 
\big( r + \alpha \big) - d_{u_0,u_{3,4}}
=
\big( r + \alpha \big) - \left( r-\left(\frac{3 \kappa -2}{12} \right) \left( 6\,\gadd_{\mathrm{worst}}(G) + 2 \right) \right)
\\
=
\alpha + \left(\frac{3 \kappa -2}{12} \right) \left( 6\,\gadd_{\mathrm{worst}}(G) + 2 \right)
\end{multline*}
Thus, by Corollary~\ref{cor1} (with $\gamma=\alpha + \left(\frac{3 \kappa -2}{12} \right) \left( 6\,\gadd_{\mathrm{worst}}(G) + 2 \right)\,$), 
we get 
\[
\ell \left( \cQ \right) \geq 
2^{^{ \frac{\gamma}{6\,\gadd_{\mathrm{worst}}(G) + 2} +1 } }-1
=
2^{^{
\textstyle
\frac{\alpha} { 6\,\gadd_{\mathrm{worst}}(G) + 2 }
+
\frac{\kappa}{4} + \frac{5}{6}
}
}
-1
\]
If $\cQ$ is a $\mu$-approximate short path, then by Corollary~\ref{cor2-1}:
\begin{multline*}
           \mu
      > 
\frac
{
   2^{
   \textstyle
   \frac{ \gamma }
   {
   6\,\gadd_{\mathrm{worst}}(G) + 2
   }
   }
}
{
   { 12\,\gamma }
   - 
              \Big( \, 24 + \mathrm{o} (1) \, \Big)
              \,
              \Big( \, 6\,\gadd_{\mathrm{worst}}(G) + 2 \, \Big)
}
- \frac { 1 } { 3 }
\\
=
\frac
{
   2^{
   \textstyle
\frac{\alpha} { 6\,\gadd_{\mathrm{worst}}(G) + 2 }
+
\frac{\kappa}{4} - \frac{1}{6}
   }
}
{
   { 12\,\alpha + \big( 3\,\kappa -26 - \mathrm{o} (1) \, \big)\, \big(6\,\gadd_{\mathrm{worst}}(G) + 2 \big)   }
}
- \frac { 1 } { 3 }
\end{multline*}
If $\cQ$ is a $\eps$-additive-approximate short path, then by Corollary~\ref{cor2-1}:
\begin{gather*}
\eps
\,>\, 
\frac{
2^{
\frac{
\gamma
}
{
6\,\gadd_{\mathrm{worst}}(G) + 2
}
}
}
{
48\,\gadd_{\mathrm{worst}}(G) \,+ \frac{17}{2}
}
\,-\,
\log_2 \left( 48\,\gadd_{\mathrm{worst}}(G) + 16 \right) 
\\
\hspace*{0.3in}
=
\frac{
2^{
   \textstyle
\frac{\alpha} { 6\,\gadd_{\mathrm{worst}}(G) + 2 }
+
\frac{\kappa}{4} - \frac{1}{6}
}
}
{
48\,\gadd_{\mathrm{worst}}(G) \,+ \frac{17}{2}
}
\,-\,
\log_2 \left( 48\,\gadd_{\mathrm{worst}}(G) + 16 \right) 
\end{gather*}



\section{Lemma~\ref{lem4}}
\label{lem4-sec}

\begin{lemma}
\label{lem4}
{\small\em\bf (equivalence of strong and weak domination; see \FI{f12-prev} for a visual illustration)}
If $\lambda\geq \left( 6\,\gadd_{\mathrm{worst}}(G)+2 \right) \log_2 n$ then
\vspace*{-5pt}
\[
\hspace*{-0.0in}
\mathfrak{M}_{u,\rho,\lambda}
\stackrel{\mathrm{def}}{=\joinrel=}
\eX{
\left.
\hspace*{-0.15in}
\text{
\begin{tabular}{p{1.4in}}
\small 
number of pairs of nodes $v,y$ such that $v,y$ is {\bf weakly} $(\rho,\lambda)$-dominated by $u$
\end{tabular}
}
\hspace*{-0.1in}
\right| 
\hspace*{-0.1in}
\text{
\begin{tabular}{p{1in}}
\small 
$v$ is selected uniformly randomly from $\cup_{\rho\,<\,j\,\leq\, \lambda} \mathcal{B}_j(u)$
\end{tabular}
}
\hspace*{-0.1in}
}
\]
\vspace*{-10pt}
\[
\hspace*{28pt}
\pmb{=}
\,\,
\eX{
\left.
\hspace*{-0.15in}
\text{
\begin{tabular}{p{1.4in}}
\small 
number of pairs of nodes $v,y$ such that $v,y$ is {\bf strongly} $(\rho,\lambda)$-dominated by $u$
\end{tabular}
}
\hspace*{-0.1in}
\right| 
\hspace*{-0.1in}
\text{
\begin{tabular}{p{1in}}
\small 
$v$ is selected uniformly randomly from $\cup_{\rho\,<\,j\,\leq\, \lambda} \mathcal{B}_j(u)$
\end{tabular}
}
\hspace*{-0.1in}
}
\]
\end{lemma}

\begin{proof}
Suppose that $v,y$ is weakly $(\rho,\lambda)$-dominated by $u$, \IE, 
there exists a shortest path 
$v \!\stackrel{\cP}{\leftrightsquigarrow}\! y$
between $v,y\in\mathcal{B}_{\rho+\lambda}(u)$
such that for some node $v'\in v\!\stackrel{\cP}{\leftrightsquigarrow}\!y$ we have $v'\in\mathcal{B}_{\rho}(u)$.                                                  
Let $v \!\stackrel{\cQ}{\leftrightsquigarrow}\! y$
be any other path between $v$ and $y$ that does {\em not} contain a node
from $\mathcal{B}_{\rho}(u)$.
Then, by Corollary~\ref{cor2}{\bf (i)} (with $\kappa=4$) we have 
\[
\ell \left( \cQ \right) \geq 
2^{^{
\textstyle
\frac{\lambda} { 6\,\gadd_{\mathrm{worst}}(G) + 2 }
+
\frac{11}{6}
}
}
-1
\geq
2^{\log_2 n + \frac{11}{6}} -1
> n-1
\]
which contradicts the obvious bound $\ell \left( \cQ \right)<n$.
Thus, {\em no} such path $\cQ$ exists and 
$v,y$ is strongly $(\rho,\lambda)$-dominated by $u$.
\end{proof}


\begin{acknowledgments}
B. DasGupta and N. Mobasheri were supported by NSF grants IIS-1160995. 
R. Albert was supported by NSF grants IIS-1161007 and PHY-1205840.
\end{acknowledgments}


\begin{figure*}
\begin{center}
{\bf\Huge Supplemental Information}
\end{center}
\end{figure*}

\clearpage

\renewcommand{\arraystretch}{1.2}
\begin{table*}
\caption{\label{L1}Details of $11$ biological networks studied}
\begin{tabular}{llccc}
\\
[-0.1in]
\multicolumn{1}{c}{name} & \multicolumn{1}{c}{brief description} & \multicolumn{1}{c}{\# nodes} & \multicolumn{1}{c}{\# edges} & \multicolumn{1}{c}{reference}
\\
\hline
\begin{tabular}{ll} 
{\bf 1}. & 
{\em E. coli} transcriptional 
\end{tabular}
& 
\begin{tabular}{l} 
{\em E. coli} transcriptional regulatory network of
\\
direct regulatory interactions between transcription 
\\
factors and the genes or operons they regulate
\end{tabular}
&
311
& 
451
&
\cite{net1} 
\\
\hline
\begin{tabular}{ll} 
{\bf 2}. & 
Mammalian signaling 
\end{tabular}
& 
\begin{tabular}{l} 
Mammalian network of signaling pathways and 
\\
cellular machines in the hippocampal CA1 neuron
\end{tabular}
&
512
& 
1047
&
\cite{net2} 
\\
\hline
\begin{tabular}{ll} 
{\bf 3}. & 
{\em E. coli} transcriptional 
\end{tabular}
& 
\begin{tabular}{l} 
{\em E. coli} transcriptional regulatory network of
\\
direct regulatory interactions between transcription 
\\
factors and the genes or operons they regulate
\end{tabular}
&
418
& 
544
&
$\pmb{\sharp}$
\\
\hline
\begin{tabular}{ll} 
{\bf 4}. & 
T-LGL signaling 
\end{tabular}
& 
\begin{tabular}{l} 
Signaling network inside cytotoxic T cells in the context of 
\\
the disease T cell large granular lymphocyte leukemia
\end{tabular}
&
58
& 
135
&
\cite{net6} 
\\
\hline
\begin{tabular}{ll} 
{\bf 5}. & 
{\em S. cerevisiae} 
\\
&
transcriptional 
\end{tabular}
& 
\begin{tabular}{l} 
{\em S. cerevisiae} transcriptional regulatory network 
\\
showing interactions between transcription factor 
\\
proteins and genes
\end{tabular}
&
690
& 
1082
&
\cite{net8} 
\\
\hline
\begin{tabular}{ll} 
{\bf 6}. & 
{\em C. elegans} 
metabolic 
\end{tabular}
& 
\begin{tabular}{l} 
The network of biochemical reactions in {\em C. elegans} metabolism
\end{tabular}
&
453
& 
2040
&
\cite{net10-1} 
\\
\hline
\begin{tabular}{ll} 
{\bf 7}. & 
Drosophila 
\\
&
segment polarity
\\
& 
(6 cells)
\end{tabular}
& 
\begin{tabular}{l} 
1-dimensional 6-cell version of the gene regulatory 
\\
network among products of the segment polarity 
\\
gene family that plays an important role in the
\\
embryonic development of {\em Drosophila melanogaster}
\end{tabular}
&
78
& 
132
&
\cite{Odell:2000}
\\
\hline
\begin{tabular}{ll} 
{\bf 8}. & 
ABA signaling
\end{tabular}
& 
\begin{tabular}{l} 
Guard cell signal transduction network for 
\\
abscisic acid (ABA) induced stomatal closure in plants
\end{tabular}
&
55
& 
88
&
\cite{LAA06}
\\
\hline
\begin{tabular}{ll} 
{\bf 9}. & 
Immune response 
\\
&
network
\end{tabular}
& 
\begin{tabular}{l} 
Network of interactions among immune cells and pathogens 
\\
in the mammalian immune response against two bacterial species 
\end{tabular}
&
18
& 
42
&
\cite{Thakar07}
\\
\hline
\begin{tabular}{ll} 
{\bf 10}. & 
T cell receptor 
\\
&
signaling
\end{tabular}
& 
\begin{tabular}{l} 
Network for T cell activation mechanisms after engagement 
\\
of the TCR, the CD$_4$/CD8 co-receptors and CD\scalebox{0.8}{2}8.
\end{tabular}
&
94
& 
138
&
\cite{julio07}
\\
\hline
\begin{tabular}{ll} 
{\bf 11}. & 
Oriented yeast PPI
\end{tabular}
& 
\begin{tabular}{l} 
An oriented version of an unweighted PPI network constructed 
\\
from {\em S. cerevisiae} interactions in the BioGRID database 
\end{tabular}
&
786
& 
2445
&
\cite{net11}
\\
\hline
\\
\multicolumn{5}{c}{
$^{\textstyle\pmb{\sharp}}\,$Updated version of the network in~\cite{net1}; see \url{www.weizmann.ac.il/mcb/UriAlon/Papers/networkMotifs/coli1_1Inter_st.txt}.
}
\end{tabular}
\end{table*}
\renewcommand{\arraystretch}{1}


\renewcommand{\arraystretch}{1.2}
\begin{table*}
\caption{\label{L2}Details of $9$ social networks studied}
\begin{tabular}{lllccc}
\\
[-0.1in]
\multicolumn{1}{c}{name} & \multicolumn{1}{c}{brief description} & \multicolumn{1}{c}{type} & \multicolumn{1}{c}{\# nodes} & \multicolumn{1}{c}{\# edges} & \multicolumn{1}{c}{reference}
\\
\hline
\begin{tabular}{ll} 
{\bf 1}. & 
Dolphin social 
\\ 
& network 
\end{tabular}
& 
\begin{tabular}{l} 
Social network of frequent associations between 
\\
$62$ dolphins in a community living off Doubtful 
\\
Sound in New Zealand
\end{tabular}
&
\begin{tabular}{c} 
undirected, 
\\
unweighted
\end{tabular}
&
62
& 
160
&
\cite{LSBHSD03} 
\\
\hline
\begin{tabular}{ll} 
{\bf 2}. & 
American 
\\ 
& 
College Football 
\end{tabular}
& 
\begin{tabular}{l} 
Network of American football games between 
\\
Division IA colleges during the regular Fall 2000 season 
\end{tabular}
&
\begin{tabular}{c} 
undirected, 
\\
unweighted
\end{tabular}
&
115
& 
612
&
\cite{GN02} 
\\
\hline
\begin{tabular}{ll} 
{\bf 3}. & 
Zachary Karate 
\\ 
& 
Club
\end{tabular}
& 
\begin{tabular}{l} 
Network of friendships between 34 members 
\\
of a karate club at a US university in the 1970s
\end{tabular}
&
\begin{tabular}{c} 
undirected, 
\\
unweighted
\end{tabular}
&
34
& 
78
&
\cite{Z77} 
\\
\hline
\begin{tabular}{ll} 
{\bf 4}. & 
Books about 
\\ 
& 
US politics 
\end{tabular}
& 
\begin{tabular}{l} 
Network of books about US politics published 
\\
around the time of the 2004 presidential 
\\
election and sold by the online bookseller 
\\
\url{amazon.com}; edges between books represent 
\\
frequent copurchasing of books by the same buyers. 
\end{tabular}
&
\begin{tabular}{c} 
undirected, 
\\
unweighted
\end{tabular}
&
105
& 
442
&
$\pmb{\ddagger}$
\\
\hline
\begin{tabular}{ll} 
{\bf 5}. & 
Sawmill 
\\ 
& 
communication 
\\
&
network 
\end{tabular}
& 
\begin{tabular}{l} 
A communication network within a small enterprise: 
\\
a sawmill. All employees were asked to indicate the 
\\
frequency with which they discussed work matters 
\\
with each of their colleagues on a five-point scale 
\\
ranging from less than once a week to several times
\\
a day. Two employees were linked in the 
network 
\\
if they rated their contact as three or more. 
\end{tabular}
&
\begin{tabular}{c} 
undirected, 
\\
unweighted
\end{tabular}
&
36
& 
62
&
\cite{MM97} 
\\
\hline
\begin{tabular}{ll} 
{\bf 6}. & 
Jazz Musician 
\\ 
&
network 
\end{tabular}
& 
\begin{tabular}{l} 
A social network of Jazz musicians
\end{tabular}
&
\begin{tabular}{c} 
undirected, 
\\
unweighted
\end{tabular}
&
198
& 
2742
&
\cite{GD03} 
\\
\hline
\begin{tabular}{ll} 
{\bf 7}. & 
Visiting ties 
\\ 
&
in San Juan 
\end{tabular}
& 
\begin{tabular}{l} 
Network for visiting relations between families 
living 
\\
in farms in the neighborhood San Juan Sur, 
\\
Costa Rica, 1948
\end{tabular}
&
\begin{tabular}{c} 
undirected, 
\\
unweighted
\end{tabular}
&
75
& 
144
&
\cite{LMCL53} 
\\
\hline
\begin{tabular}{ll} 
{\bf 8}. & 
World Soccer
\\ 
&
Data,
\\
&
Paris 1998
\end{tabular}
& 
\begin{tabular}{l} 
Members of the 22 soccer teams which participated 
\\
in the World Championship in Paris in 1998 had 
\\
contracts in 35 countries. Counts of which team 
\\
exports how many players to which country are
\\
used to generate this network.
\end{tabular}
&
\begin{tabular}{c} 
directed, 
\\
weighted
\end{tabular}
&
35
& 
118
&
$\pmb{\dagger}$
\\
\hline
\begin{tabular}{ll} 
{\bf 9}. & 
Les Miserables
\end{tabular}
& 
\begin{tabular}{l} 
Network of co-appearances of characters in Victor 
\\
Hugo's novel ``Les Miserables''.  Nodes represent 
\\
characters as indicated by the labels and edges 
\\
connect any pair of characters that appear in the 
\\
same chapter of the book. The weights on the 
\\
edges are the number of such coappearances. 
\end{tabular}
&
\begin{tabular}{c} 
undirected, 
\\
weighted
\end{tabular}
&
77
& 
251
&
\cite{Knu93} 
\\
\hline
\\
&
\multicolumn{5}{l}{
\begin{tabular}{l}
$^{\textstyle\pmb{\ddagger}}\,$V. Krebs, unpublished manuscript, found on Krebs' website \url{www.orgnet.com}.
\end{tabular}
}
\\
[0.1in]
&
\multicolumn{5}{l}{
\begin{tabular}{l}
$^{\textstyle\pmb{\dagger}}\,$Dagstuhl seminar: {\em Link Analysis and Visualization}, Dagstuhl 1-6, 2001.
\\
\hspace*{0.1in}(see \url{http://vlado.fmf.uni-lj.si/pub/networks/data/sport/football.htm})
\end{tabular}
}
\end{tabular}
\end{table*}
\renewcommand{\arraystretch}{1}

\begin{table*}
\begin{center}
\bf 
\large
Biological details of source, target and central nodes ($u_{\mathrm{source}}$, $u_{\mathrm{target}}$ and $u_{\mathrm{central}}$)
used in Table~\ref{KO-prev} and Table~\ref{KO}
\end{center}

\section*{Network 1:  E. coli transcriptional}

\vspace*{0.2in}
\hspace*{0.2in}
\begin{tabular}{r c p{5in}}
\hline
\\
\multicolumn{1}{c}{Node name} & 
\multicolumn{1}{c}{Node type} & 
\multicolumn{1}{c}{Details}
\\
\\
\hline
\\
fliAZY & \large$u_{\mathrm{source}}$ & Contains fliA gene (sigma factor), fliZ (possible cell-density responsive regulator of sigma) and fliY (periplasmic cystine-binding protein)
\\
[0.1in]
fecA & \large$u_{\mathrm{source}}$ & Ferric citrate, outer membrane receptor
\\
[0.1in]
arcA & \large$u_{\mathrm{target}}$ & Aerobic respiration control, transcriptional dual regulator
\\
[0.1in]
aspA & \large$u_{\mathrm{target}}$ & Component of aspartate ammonia-lyase
\\
[0.1in]
crp & \large$u_{\mathrm{central}}$ & Component of CRP transcriptional dual regulator (DNA-binding transcriptional dual regulator)
\\
[0.1in]
CaiF & \large$u_{\mathrm{central}}$ & DNA-binding transcriptional activator
\\
[0.1in]
sodA & \large$u_{\mathrm{central}}$ & Component of superoxide dismutases that catalyzes the dismutation of superoxide into oxygen and hydrogen peroxide
\\
\\
\hline
\\
\\
\\
\end{tabular}

\section*{Network 4: T-LGL signaling network}

\vspace*{0.2in}
\hspace*{0.2in}
\begin{tabular}{r c p{5in}}
\hline
\\
\multicolumn{1}{c}{Node name} & 
\multicolumn{1}{c}{Node type} &
\multicolumn{1}{c}{Details}
\\
\\
\hline
\\
PDGF & \large$u_{\mathrm{source}}$ & Platelet-derived growth factor is one of the numerous growth factors, or proteins that regulates cell growth and division.
\\
[0.1in]
IL15 & \large$u_{\mathrm{source}}$ & Interleukin 15 is a cytokine.
\\
[0.1in]
Stimuli &  \large$u_{\mathrm{source}}$ & Antigen Stimulation
\\
[0.1in]
apoptosis & \large$u_{\mathrm{target}}$ & process of programmed cell death
\\
[0.1in]
IL2 & \large$u_{\mathrm{central}}$ & Interleukin 2 is a cytokine signaling molecule in the immune system
\\
[0.1in]
Ceramide & \large$u_{\mathrm{central}}$ & A waxy lipid molecule within the cell membrane which can participate in variety of cellular signaling like proliferation and apoptosis
\\
[0.1in]
GZMB & \large$u_{\mathrm{central}}$ & A serine proteases that is released within cytotoxic T cells and natural killer cells to induce apoptosis within virus-infected cells, thus destroying them
\\
[0.1in]
NFKB &  \large$u_{\mathrm{central}}$ & nuclear factor kappa-light-chain-enhancer of activated B cells, a protein complex that controls the transcription of DNA
\\
[0.1in]
MCL1 &  \large$u_{\mathrm{central}}$ & Induced myeloid leukemia cell differentiation protein Mcl-1
\\
\\
\hline
\\
\\
\\
\end{tabular}

\end{table*}

\end{document}